\documentclass[superscriptaddress, reprint, amsmath, amssymb, aps, pra, floatfix]{revtex4-2}

\usepackage[colorlinks=true]{hyperref}
\usepackage{graphicx}
\usepackage{dcolumn}
\usepackage{bm}
\usepackage{orcidlink}
\usepackage{physics}
\usepackage{color}
\usepackage{amsthm}
\usepackage{multirow}
\newtheorem{thm}{Theorem}

\definecolor{purple}{rgb}{0.63,0,1}

\definecolor{pink}{rgb}{1,0,0.9}

\begin{document}

\title{Spectral gaps of two- and three-dimensional many-body quantum systems in the thermodynamic limit}

\author{Illya V. Lukin\orcidlink{0000-0002-8133-2829}}
\email{illya.lukin11@gmail.com}
\affiliation{Karazin Kharkiv National University, Svobody Square 4, 61022 Kharkiv, Ukraine}
\affiliation{Akhiezer Institute for Theoretical Physics, NSC KIPT, Akademichna 1, 61108 Kharkiv, Ukraine}

\author{Andrii G. Sotnikov\orcidlink{0000-0002-3632-4790}}
\affiliation{Karazin Kharkiv National University, Svobody Square 4, 61022 Kharkiv, Ukraine}
\affiliation{Akhiezer Institute for Theoretical Physics, NSC KIPT, Akademichna 1, 61108 Kharkiv, Ukraine}

\author{Jacob M. Leamer\orcidlink{0000-0002-2125-7636}}
\affiliation{Department of Physics and Engineering Physics, Tulane University, 6823 St. Charles Ave., New Orleans, LA 70118, USA}\affiliation{Center for Computing Research, Sandia National Laboratories, Albuquerque, New Mexico 87185, USA}

\author{Alicia B. Magann\orcidlink{0000-0002-1402-3487}}
\affiliation{Center for Computing Research, Sandia National Laboratories, Albuquerque, New Mexico 87185, USA}

\author{Denys I. Bondar\orcidlink{0000-0002-3626-4804}}\email{dbondar@tulane.edu}
\affiliation{Department of Physics and Engineering Physics, Tulane University, 6823 St. Charles Avenue, New Orleans, Luisiana 70118, USA}

\begin{abstract}

We present an expression for the spectral gap, opening up new possibilities for performing and accelerating spectral calculations of quantum many-body systems. We develop and demonstrate one such possibility in the context of tensor network simulations. Our approach requires only minor modifications of the widely used simple update method and is computationally lightweight relative to other approaches. We validate it by computing spectral gaps of the 2D and 3D transverse-field Ising models and find strong agreement with previously reported perturbation theory results.
\end{abstract}

\maketitle

\section{Introduction}
Understanding the behavior of quantum many-body systems in 2D and 3D lies at the heart of contemporary physics. One of the key enigmas in this domain is the computation of the spectral gap, i.e., the energy difference between the ground state and the first excited state, as the spectral gap serves as a key parameter characterizing the intrinsic properties of various quantum phenomena, including superconductivity, superfluidity, spin liquids~\cite{han_fractionalized_2012}, and quantum annealing~\cite{albash_adiabatic_2018}.

Here, we introduce a numerical method designed to determine the spectral gap in quantum systems. This approach is computationally lightweight and seamlessly integrates with widely used imaginary time propagators, enabling its broad applicability across diverse quantum systems. The computational efficiency of the method positions it as a versatile tool to unravel the essential characteristics of quantum systems. 

To validate this numerical approach, we apply it to the well-known 2D transverse-field Ising model. The obtained results exhibit excellent agreement with existing analytical solutions derived through perturbation theory, establishing the reliability of the method. We then expand our calculations to encompass the 1D Haldane chain, as well as the 3D transverse-field Ising model. This latter challenge was previously considered beyond the reach of existing computational techniques, and as a consequence, estimates of the spectral gap in 3D have, until now, been confined to series expansions~\cite{weihong_series_1994, pishtshev_new_2007}. We numerically benchmark the results of our method against these results, showcasing the broad applicability and computational prowess of this method in addressing intricate quantum systems. 

\section{Theoretical approach}
\subsection{Main theorem}

We begin by stating the main analytical result.
\begin{thm}
    Let $H$ be a self-adjoint Hamiltonian such that its spectral decomposition reads 
   \begin{align}
       H = \sum_{n=0}^N E_n\Pi_n, \label{eq:h_decomp}
   \end{align} where $\Pi_n$ are orthogonal projectors, i.e., $H \Pi_n = E_n \Pi_n$, $\Pi_n\Pi_m = \delta_{mn}\Pi_n$, and $E_0 < E_1 < E_2 < \cdots < E_N$ are eigenenergies.  Also let \begin{align}\label{eq:imag_time_prop}
       \ket{\phi(\tau)} = \mathcal{N} e^{-\tau H}\ket{\phi(0)} = \mathcal{N}\sum_{n=0}^N e^{-\tau E_n}\Pi_n\ket{\phi(0)},
   \end{align}
   denote a state $\ket{\phi(0)}$ propagated in imaginary time $\tau$ by preserving the norm. Here, $\mathcal{N}$ is the normalization constant ensuring that $\braket{\phi(\tau)}{\phi(\tau)} = 1$ for all $\tau$.
   
   For a self-adjoint observable $O$, if $\bra{\phi(0)}\Pi_0 O\Pi_1\ket{\phi(0)}$ is not a real number, i.e.,
   \begin{align}\label{eq:condition_for_first_spectral_gap}
        \bra{\phi(0)}\Pi_0 O\Pi_1\ket{\phi(0)} \neq \bra{\phi(0)}\Pi_1 O\Pi_0\ket{\phi(0)},    
   \end{align}
   then as $\tau \to \infty$,
   \begin{align}\label{eq:first_spectral_gap}
       \ln\big|\bra{\phi(\tau)} [H,O] \ket{\phi(\tau)} \big|
        &= -\tau \Delta + \mathcal{O}(1), \\
        \Delta &= E_1-E_0. \notag 
   \end{align}
    Furthermore, if an observable $O$ is such that $ \bra{\phi(0)}\Pi_0 O\Pi_1\ket{\phi(0)}$ is real, but  $\bra{\phi(0)}\Pi_0 O\Pi_2\ket{\phi(0)}$ is non real, i.e.,
\begin{align}\label{eq:condition_for_second_spectral_gap}
        \bra{\phi(0)}\Pi_0 O\Pi_1\ket{\phi(0)} = \bra{\phi(0)}\Pi_1 O\Pi_0\ket{\phi(0)} \mbox{ and}\notag\\
        \bra{\phi(0)}\Pi_0 O\Pi_2\ket{\phi(0)} \neq \bra{\phi(0)}\Pi_2 O\Pi_0\ket{\phi(0)},
\end{align}
then as $\tau \to \infty$,
\begin{align}\label{eq:second_spectral_gap}
       \ln\big|\bra{\phi(\tau)} [H,O] \ket{\phi(\tau)} \big|
        = -\tau(E_2-E_0) + \mathcal{O}(1).
\end{align}
\end{thm}
\begin{proof}
We begin 
\begin{align}\label{eq:starting_point}
    \bra{\phi(0)} & e^{-\tau H}[H,O] e^{-\tau H}\ket{\phi(0)}  \notag\\
    =& \sum_{l,k = 0}^N e^{-\tau (E_l + E_k) }\bra{\phi(0)} \Pi_l (HO - OH) \Pi_k \ket{\phi(0)} \notag\\
    =& \sum_{l,k = 0}^N e^{-\tau (E_l + E_k) } (E_l - E_k) \bra{\phi(0)} \Pi_l O \Pi_k \ket{\phi(0)}.
\end{align}
Under the condition~\eqref{eq:condition_for_first_spectral_gap}, we get
\begin{align}\label{eq:Phi0adAsympt}
     \bra{\phi(0)} & e^{-\tau H}[H,O] e^{-\tau H}\ket{\phi(0)}  \notag\\
    =& e^{-\tau (E_0 + E_1) } \Big[(E_0 - E_1) \bra{\phi(0)} \Pi_0 O \Pi_1 \ket{\phi(0)} \notag\\
    & + (E_1 - E_0) \bra{\phi(0)} \Pi_1 O \Pi_0 \ket{\phi(0)} \Big] \notag\\
    & \qquad + \mathcal{O}\left( e^{-\tau (E_0 + E_2) } \right) \notag\\
    =& (E_1 - E_0) e^{-\tau (E_0 + E_1) } \bra{\phi(0)} \Pi_1 O \Pi_0 \notag\\
    &- \Pi_0 O \Pi_1 \ket{\phi(0)}\left[ 1 + \mathcal{O}\left( e^{-\tau (E_2 - E_1) } \right)\right].
\end{align}

Condition~\eqref{eq:condition_for_first_spectral_gap} also implies that the state $\ket{\phi(0)}$ has a population in the ground state, i.e.,
\begin{align}
    & \bra{\phi(0)}\Pi_1 O\Pi_0\ket{\phi(0)} \neq 0 \Longrightarrow
    \Pi_0\ket{\phi(0)} \neq \ket{\varnothing} \Longrightarrow \notag\\
    & \bra{\phi(0)} \Pi_0 \ket{\phi(0)} \neq 0,
\end{align}
where $\ket{\varnothing}$ denotes the null vector. With this in mind, we can estimate the normalization constant in Eq.~\eqref{eq:imag_time_prop} as
\begin{align}\label{eq:NAsympt}
    \mathcal{N}^2 &= \left[\bra{\phi(0)} e^{-2\tau H} \ket{\phi(0)} \right]^{-1} \notag\\
        &= \left[e^{-2\tau E_0} \bra{\phi(0)} \Pi_0 \ket{\phi(0)} + \mathcal{O}\left( e^{-2\tau E_1}\right) \right]^{-1} \notag\\
        &= e^{2\tau E_0} \mathcal{O}(1), \qquad \tau \to \infty.
\end{align}
Using Eqs.~\eqref{eq:Phi0adAsympt} and \eqref{eq:NAsympt}, we obtain
\begin{align}
    \ln & \big|\bra{\phi(\tau)} [H,O] \ket{\phi(\tau)} \big| \notag\\
    &= \ln \mathcal{N}^2 + \ln \left| \bra{\phi(0)}  e^{-\tau H}[H,O] e^{-\tau H}\ket{\phi(0)} \right| \notag\\
    & = 2 \tau E_0 + \mathcal{O}(1) -\tau (E_0 + E_1) \notag\\
    &= -\tau (E_1 - E_0) + \mathcal{O}(1),
\end{align}
thereby completing the proof of Eq.~\eqref{eq:first_spectral_gap}.

Under the condition~\eqref{eq:condition_for_first_spectral_gap}, the leading term in expansion~\eqref{eq:starting_point} reads
\begin{align}
    \bra{\phi(0)} & e^{-\tau H}[H,O] e^{-\tau H}\ket{\phi(0)} = \mathcal{O}(1) e^{-\tau (E_0 + E_2) },
\end{align}
whence Eq.~\eqref{eq:second_spectral_gap} follows.
\end{proof}

We note that Theorem~1 is a modification of the theorem in Ref.~\cite{leamer2023spectral} and that previously, other works have targeted the development of similar, albeit approximate, spectral gap expressions (see, e.g., Refs.~\cite{koh_effects_2016, hlatshwayo2023quantum}). It has also been long recognized that the energies of the low lying excited states can be extracted via the multi-exponential fitting of a correlation function (see, e.g., Refs.~\cite[Sec.VII.C]{ceperley_calculation_1988} and \cite{caffarel_development_1988}). However, such a fit is numerically unstable. Importantly, Eq.~\eqref{eq:first_spectral_gap} does not suffer from this drawback.

In the thermodynamic limit, quantum systems (including the 2D and 3D transverse-field Ising models studied below) often have a continuum spectrum in addition to discrete energies. Equation~\eqref{eq:first_spectral_gap} is also applicable in such systems; $E_1$ is either the first excited state, if the system has at least two discrete energy levels, or the infimum of the continuous spectrum, if there is only one bound state.

A practical application of Eq.~\eqref{eq:first_spectral_gap} to estimate the spectral gap requires a numerical scheme for imaginary time evolution to calculate the ground state. There are numerous computational techniques that could be considered for this purpose, e.g., tensor networks, quantum Monte Carlo \cite{RevModPhys.73.33}, and even quantum computing-based methods~\cite{motta2020determining, russo_evaluating_2021, gnatenko_detection_2022, gnatenko_energy_2022, stroeks_spectral_2022, leamer2024quantum}, and exploration of each would constitute valuable future work. 

Here, we specifically consider imaginary time evolution within tensor networks. Tensor networks, as detailed in~\cite{Orus2019review, 2021_Cirac, 2023_Banuls, Okunishi_2022}, are an apt wave-function ansatz for capturing the entanglement structure of ground states. This approach is particularly effective in 2D~\cite{2021_Bruognolo} and in 1D using time-evolving block decimation~\cite{2004_Vidal}. 

For noncritical phenomena in 1D, the matrix product states (MPS) ansatz~\cite{SCHOLLWOCK201196}, integral to the density matrix renormalization group (DMRG)~\cite{1992_White}, is widely utilized. Its extension to 2D and 3D systems is achieved through projected entangled pair states (PEPS)~\cite{2008_Verstraete, 2021_Cirac, 2021_Bruognolo} and iPEPS--the infinite-size limit of PEPS. We note that iPEPS have proven effectiveness in capturing strong correlations in magnetic systems~\cite{2012_Corboz, 2013_Corboz}, fermionic models~\cite{2014_Corboz, 2015_Corboz}, topological spin liquids~\cite{2022_Hasik}, finite temperature systems~\cite{2020_Jahromi, 2015_Czarnik}, time evolution~\cite{2019_Czarnik}, and excited states~\cite{2019_Vanderstraten_PEPS}.

Here, we apply the simple update method~\cite{2008_Jiang, 2021_Bruognolo} for the imaginary time evolution of iPEPS and observables calculations~\cite{2019_Jahromi, 2023_Tindall}. We calculate the spectral gap via Eq.~\eqref{eq:first_spectral_gap} by tracing an expectation value during the imaginary time evolution.

\subsection{Transverse-field Ising model in 2D and 3D}\label{subsec:2b}

The Hamiltonian has the following form:
\begin{equation}
    H = -J\sum_{\langle ij \rangle} \sigma_{i}^{z} \sigma_{j}^{z} - g \sum_{i} \sigma_{i}^{x},
\end{equation}
where $i$ and $j$ label sites on the square or cubic lattice, $\langle ij \rangle$ restricts summation over the nearest neighbor pairs, and $\sigma^{x}_{i}$ and $ \sigma^{z}_{i}$ are the conventional Pauli matrices acting on the site $i$. We focus on the case of ferromagnetic coupling ($J > 0$).
At zero temperature the system can be in two different phases: symmetry-broken ferromagnetic phase at small $g/J$ and disordered paramagnetic phase at large $g/J$. In 2D, the phase transition occurs at $g/J = 3.04438(2)$ \cite{2002_Blote}, while in 3D, the estimates are less precise,  $g/J \approx 5.29$ \cite{2013_latorre}.  We will use the following parametrization:  in the ferromagnetic phase, we set $J=1$ and measure the gap in units of $J$. In the paramagnetic phase, we instead set $g=1$.

The excitations in the model have the following structure: in the ferromagnetic phase, the vacuum is either all spins are aligned either up or down.  An excitation in this phase is created by flipping the spin at a single site. In 2D, flipping one spin creates four domain walls. As a result, in the case of $g = 0$, the spectral gap is $\Delta = 8J$. The series expansion up to $\mathcal{O}(g^{20})$ for the spectral gap in the ferromagnetic phase in 2D can be found in Ref.~\cite{Oitmaa_1991}.

In the noninteracting case where $J = 0$, all spins are aligned along the positive $x$ direction in the ground state. The elementary excitation corresponds to the spin flip on one site, with the excitation energy $\Delta = 2 g$, or simply $\Delta = 2$ due to the chosen convention in the paramagnetic phase. The series expansion up to $\mathcal{O}(J^{13})$ for the spectral gap in the paramagnetic phase can be found in Ref.~\cite{He_1990} and Ref.~\cite{weihong_series_1994} in 2D and 3D, respectively.

It is important to note that in both the ferromagnetic and paramagnetic phases, the excitation can be obtained from the ground state by applying the local operator $\sigma^{y}_{i}$, since this operator reverses the magnetization along both the $z$ and $x$ axes. More generally, the operator $\sum_{i} \sigma_{i}^{y}$ is expected to exhibit a nonzero overlap between the ground state and the first excited state, regardless of the value of $g/J$. Furthermore, since $\sigma_{i}^{y}$ contains the purely imaginary matrix elements and we take the real valued initial function $\ket{\phi(0)}$, the operator $O = \sum_{i} \sigma_{i}^{y}$ satisfies the condition~\eqref{eq:condition_for_first_spectral_gap}, thereby validating application of Eq.~\eqref{eq:first_spectral_gap} for the spectral gap.

\subsection{Gapped topological system: 1D Haldane chain}
\label{subsec:2c}

The 1D Haldane model can be viewed as the one-dimensional antiferromagnetic Heisenberg chain with spin $S=1$. The model Hamiltonian is
\begin{equation}
    H = \sum_{i} \vec{S}_{i} \cdot \vec{S}_{i+1}.
\end{equation}
It was conjectured by Haldane that this system is in the gapped phase \cite{HALDANE1983464, Haldane83}. This was later confirmed by extensive numerical calculations \cite{Nakano_2009, Golinelli_1994, White_1993, Todo_2001} and also by the solution of the Affleck--Kennedy--Lieb--Tasaki (AKLT) model~\cite{AKLT}, which is in the same phase. Furthermore, the gapped phase is nontrivial---this is the symmetry-protected topological order, which is protected simultaneously by several symmetries \cite{Pollmann_2010}. In the tensor network calculations these topological properties of the model expose themselves as degeneracies in the entanglement spectra \cite{Pollmann_2010}, which we also observe in our calculations.

For the AKLT model the structure of excitations is approximately known, since these are determined by the spin triplets on the lattice bonds. The excitations (for the AKLT model) can be created by the operator $O = \sum_{i} S_{i}^{y} S_{i+1}^{z}$ (note that the additional higher-energy excitations are also created, but they vanish quickly in imaginary-time evolution). We employ the same $O$ for the gap extraction in the Haldane chain, since the Haldane chain and AKLT models are continuously connected. 

\subsection{Spectral gap calculations via tensor networks}
\label{subsec:2d}

There are several different ways to calculate the gap and excitation spectra with tensor networks. The easiest approach is to add a penalizing ground state projector to the Hamiltonian and then re-run the ground-state calculation~\cite{2012_Stoudenmire, Banuls2013}. This approach works well for 1D systems, where the DMRG optimization can be efficiently employed. Another strategy is based on the tensor-network approach for the excitation ansatz~\cite{2012_Haegeman, 2013_Haegeman, 2018_Vanderstraeten, 2019_Vanderstraeten, 2013_Draxler, 2020_Vanderstraeten_spinons}.
This approach allows to calculate not only the gap and lowest excited state, but also the dispersion relation, certain two-particle bound states~\cite{2015_Vanderstraeten_bound}, the scattering matrix~\cite{2014_Vanderstraeten}, and topological excitations~\cite{2018_Zauner, 2019_Vanderstraeten}. This approach can also be applied to 2D systems on cylinders~\cite{2021_Damme}, helixes~\cite{zhang2023} or in infinitely extended 2D systems with the iPEPS excitation ansatz~\cite{2015_Vanderstraeten, 2019_Vanderstraten_PEPS, 2020_Ponsioen, 2022_Ponsioen}. There is also an approach based on the Lanzcos algorithm adaptation to the MPS~\cite{baker2023direct, 2014_Dolgov}. Note that the excitation energies can also appear in the projected Hamiltonian problems in the DMRG sweeps~\cite{2017_Chepiga} and may be connected with the ground state transfer matrix spectrum~\cite{2014_Zauner, 2023_Eberharter}. For certain class of models it was also proposed to estimate the gap by studying phase transitions upon the introduction into Hamiltonian of additional terms commuting with the Hamiltonian~\cite{2013_Garcia}. For many-body localized systems there are tailored approaches to study not only the lowest excited states but also eigenstates in the middle of the spectra~\cite{2016_Khemani, 2016_Pollmann}.

The approaches just reviewed are very effective for quasi-1D systems and for certain 2D systems within the variational iPEPS methodology. Unfortunately, for 2D systems the computations are rather demanding, since they require either complex summations of correlation functions with the corner transfer matrix renormalization group (CTMRG)~\cite{2019_Vanderstraten_PEPS, 2020_Ponsioen} or automatic differentiation through CTMRG~\cite{2022_Ponsioen}. These approaches are also currently very difficult to generalize to 3D  and complex graph structures. Here, we show that the spectral gap can be easily extracted via the simplest algorithm of tensor network optimization, which is applicable to arbitrary lattice or network structures. We illustrate this with the spectral gap estimation for the 3D transverse Ising model.

\begin{figure}
    \includegraphics[width= \linewidth]{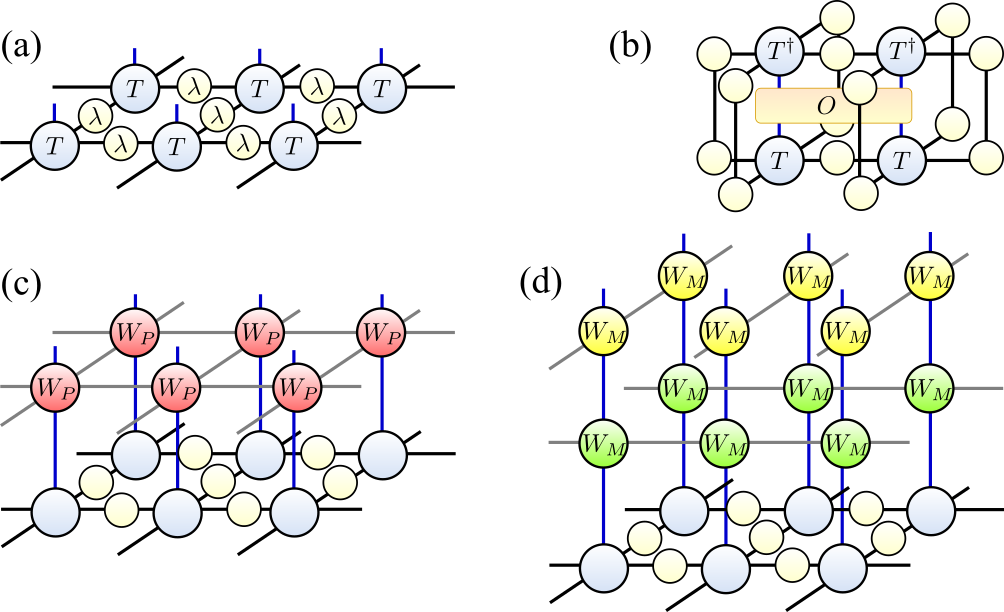}
    \caption{\label{fig:iPEPS}%
        Tensor network computations in 2D: (a) the network diagram of the iPEPS wave function, which consists of the tensors~$T$ placed on the sites of the square lattice. Each tensor has one physical index and four auxiliary virtual indices, which are geometrically contracted with the bond matrices $\lambda$. Both $T$ and $\lambda$ are in the superorthogonal canonical form~\cite{2023_Tindall, 2015_Phien, 2012_Ran}. (b) The simple update approximation to the computation of the expectation value of the two-site operator $O$~\cite{2019_Jahromi}. (c) The PEPO optimization scheme, where we apply the PEPO operator (consisting of elementary tensors $W_P$) to the iPEPS wave function. The PEPO tensors approximate the operator $\exp(-H d\tau)$ and then we truncate iPEPS with a superorthogonalisation~\cite{2015_Phien, 2015_Zaletel}. (d) The iPEPS optimization with consequent applications of horizontal and vertical MPOs (with individual tensors $W_M$). We construct MPO with the $WII$ method~\cite{2015_Zaletel} and also truncate iPEPS according to the superorthogonal canonical form.  
        }
\end{figure}

Our algorithm for calculating the spectral gap runs as follows: We employ the infinite projected entangled pair states of the bond dimension $D$ and optimize the randomly initialized iPEPS wave function (within a predefined periodic unit cell) with the imaginary time evolution (see Fig.~\ref{fig:iPEPS}). We discretize the imaginary time propagation of a step $d\tau$ and apply either the Trotterized gates representing $\exp(-H d\tau)$ or a PEPO-like approximation of the $\exp(-H d\tau)$, constructed via the $WII$ method~\cite{2015_Zaletel}, to the iPEPS. Following the application of these gates/PEPO, the bond dimension $D$ of the iPEPS wave function typically increases and must be truncated back to its original value. This truncation is performed using the simple update method, a cost-effective and widely used approach in iPEPS optimization, especially for gapped nontopological systems~\cite{2008_Jiang,2021_Bruognolo}. 

To determine the gap, it is necessary to find the operator averages of the commutator $[H, O]$ of the Hamiltonian $H$ and the observable $O$  creating excitations. For example, we employ $O = \sum_{i} \sigma^{y}_{i}$ for the Ising model, and $O = \sum_{i} S_{i}^{y} S_{i+1}^{z}$ for the Haldane model. 

The precise computation depends on the exact CTMRG contraction~\cite{baxter1968dimers,1996_Nishino, 2009_Orus, Nishino_1997} at every step of the imaginary time evolution, rendering the simple update method ineffective due to CTMRG's significantly higher computational cost compared to the simple update's truncation procedures. Fortunately, our focus is on gapped systems, where it was recently demonstrated that the simple update method, with sufficiently large bond dimensions ($D$), can accurately determine observables in these systems\cite{2019_Jahromi}. This finding led to the development of the gPEPS strategy, applying the simple update method to optimize the iPEPS wave function and to compute its observables~\cite{2019_Jahromi, 2020_Jahromi} (see also a similar discussion in~\cite{2023_Tindall}). This approach has proven successful for challenging 2D problems, some 3D systems, and for simulating the IBM kicked Ising experiment classically~\cite{2023_Patra_IBM, 2023_tindall_IBM}.

Let us provide some additional comments on different computational schemes within the simple update iPEPS methods:

(1) PEPO scheme: the optimization using the PEPO application, where PEPO approximates the exponent of the Hamiltonian $\exp(-H d\tau)$, and the iPEPS is truncated using the superorthogonalisation canonical form~\cite{2015_Phien, 2012_Ran, 2023_Tindall}.   The unit cell consists of one site. 

(2) MPO scheme: the optimization is performed by the consequent applications of horizontal and vertical MPO to the iPEPS wave function, where the MPOs encode the exponent of the Hamiltonian. The truncation is performed again using the superorthogonalisation.  This approach is faster than the PEPO scheme.  The unit cell consists again of one site.

(3) Gate scheme: the optimization is performed with the Trotterized gates. This is the most widely used scheme in practical calculations. Usually, superorthogonalization is not performed in the truncation procedure, but it can be used as an additional step between gate applications, which allows for the further stabilization of the algorithm.  The unit cell consists of the two different tensors placed in the checkerboard pattern. This method is the most unstable in terms of the gap estimation. In particular, even if the commutator decay is approximately exponential in the approach, its numerical derivative  at $d\tau\sim10^{-3}$ shows sizable fluctuations, which are absent in the two previous schemes. The gap extraction in the approach requires in this case an additional data flattening. 

Below, we employ the simple update (SU) to obtain the expectation value of the commutator $[H, O]$, where $O = \sum_{i} \sigma^{y}_{i}$ is the observable with nonzero overlap between the first excitation and ground state. The commutator expectation value is determined after each imaginary time step $d\tau$.  We then plot the logarithmic quantity $C(\tau) = \ln{|\langle [H, O] \rangle(\tau)|}$ [see also Eq.~\eqref{eq:first_spectral_gap}] and determine the interval of $\tau$ with a linear dependence. Within this interval, we interpolate the function $C(\tau) = C_{0} - \tau \Delta$, which contains the estimate of the spectral gap $\Delta$. 

Note that the proposed approach may encounter difficulties when addressing topological excitations, such as anyons in 2D systems or kinks/domain walls in 1D systems. This challenge arises because the single topological excitations generally cannot emerge from a translationally invariant tensor-network wave function through the action of a local operator. Specifically, it is not feasible to obtain the domain wall energy in the 1D Ising model inside the ferromagnetic phase. Consequently, we conclude that addressing topological excitations necessitates more complex approaches. Such methods are akin to those discussed in Ref.~\cite{2019_Vanderstraeten} for 1D systems.

Below, we present results obtained using three different schemes of the simple update iPEPS optimization: PEPO, MPO, and Trotter gates. Specifically, for the 2D system, we illustrate the iPEPS wave function, the calculation of observables, and the application of PEPO/MPO in Fig.~\ref{fig:iPEPS}. Regarding the application of Trotter gates within the simple update approach, we direct readers to Refs.~\cite{2008_Jiang,2021_Bruognolo} for further details.

\section{Results}
As an initial iPEPS wave function $\ket{\phi(0)}$, we used random product states with bond dimension $D=1$, noting that such initial states lead to better convergence. Figure~\ref{fig:Commutator_decay}(a) displays the dynamics of the commutator expectation value over imaginary time $\tau$, plotted on a logarithmic scale. Computations using all three methods (MPO, PEPO, and Trotterized gates) are shown. It is clear that all methods result in the exponential decay of the commutator (evidenced as a linear region on the graph), with nearly identical slopes, indicative of the excitation gap. However, the stability and accuracy of these methods generally vary. For instance, with the Trotter gates, for large $D$ or, surprisingly, small $d\tau$ we observe instabilities.  Such instabilities, which can occur at arbitrary $\tau$, lead to significant variations in the numerical derivatives.  Consequently, for methods based on Trotter gates, one needs to use a linear regression to  average out the numerical instabilities to extract the spectral gap, which makes it difficult to estimate the error from the time discretization. In contrast, the numerical derivatives obtained with PEPO/MPO methods are always smooth and stable [Fig.~\ref{fig:Commutator_decay}(b)], enabling direct estimation of the gap.
\begin{figure}
    \includegraphics[width= \linewidth]{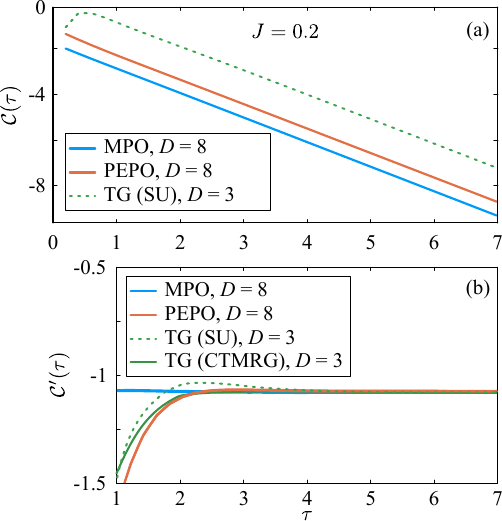}
    \caption{\label{fig:Commutator_decay}%
        (a) The decay of the expectation value ${\cal C}(\tau)=\ln|\langle [H, O] \rangle(\tau)|$ with the imaginary time $\tau$ for the 2D transverse Ising model in the paramagnetic phase with $J = 0.2$, $d\tau=0.2$, and the iPEPS bond dimension $D=8$ for PEPO/MPO and $D=3$ for the Trotter gates (TG). 
        (b) The numerical derivative of the commutator decay for the same parameters and methods, including TG scheme with CTMRG at $\chi=10$ for comparison. 
        }
\end{figure}

According to Fig.~\ref{fig:Commutator_decay}(b), the PEPO and MPO methods estimate the spectral gap $\Delta = 1.074$ (almost insensitive to the step size $d\tau$, with the calculations carried down to $d\tau = 0.002$). As discussed below, the perturbation series expansion yields $\Delta = 1.083$ [see Fig.~\ref{fig:Gap_paramagnetic}(a)]. 
For Trotter gates, $\Delta \approx 1.073$ [see Fig.~\ref{fig:Commutator_decay}(b)], but only for the specific $d\tau = 0.2$; for $d\tau < 0.2$, the results become sensitive to the step size due to significant numerical instabilities, as mentioned in Sec.~\ref{subsec:2d}.
Below, we utilize the MPO method, since it is faster than  PEPO and more reliable than the Trotter gates optimization. We have also calculated the average values with the much more accurate CTMRG approach. The results for the gap agree between the different approaches for the observable calculation, even though the average values of the observables themselves may be underestimated for the SU method of observables calculation.

Next, we benchmark this method of the gap estimation against the series expansions~\cite{He_1990}. In Fig.~\ref{fig:Gap_paramagnetic}(a), we compare the gap values computed via the simple update imaginary time evolution (using MPOs) of the iPEPS wave function with series expansions in $J$ ($g = 1$). In the region of $J$ far from the critical transition point $J_{c} \approx 0.329$, the results from the series expansions and the simple update show a strong agreement. However, in the critical region, the predictions diverge, with the simple update underestimating the transition point, while the series expansions overestimate it. The iPEPS approach becomes unreliable in this region due to the simple update's reliance on mean-field environments, which are ineffective amidst long-range correlations. To achieve more accurate results in the near-critical regime, more complex tensor-network methods should be employed: either gap extraction from the full update evolution~\cite{2008_Jordan} (complemented with CTMRG for calculating averages) or employing some type of variational evolution~\cite{2022_Dzarmaga}. Additionally, we compare these results with $D=5$ calculations, noting differences predominantly near the critical regime. It is important to note that in the critical region, gap estimation using this approach faces challenges: first, the commutator decay tends to be polynomial rather than exponential over the extended periods of $\tau$, and second, the simple update often inaccurately determines the precise position of the transition point, leading to a shifted gap prediction.

In Fig.~\ref{fig:Gap_paramagnetic}(b) we benchmark our numerical results in the ferromagnetic phase against the series expansions~\cite{Oitmaa_1991}. Unfortunately, the predictions from the series expansions are reliable only for $g < 1.5$. Within this range, our results agree well with the series expansions. We also extend our gap prediction up to $g = 2.5$, where our estimates begin to diverge from the more accurate predictions provided by the more advanced (and computationally demanding) variational iPEPS approach detailed in Ref.~\cite{2020_Ponsioen}.
\begin{figure}
    \includegraphics[width= \linewidth]{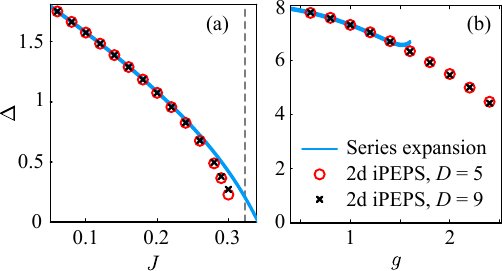} 
    \caption{\label{fig:Gap_paramagnetic}%
       (a) The gap in the paramagnetic phase obtained from the series expansion \cite[Table 1]{He_1990} up to $\mathcal{O}(J^{13})$ and from our approach to the gap estimation (with the MPO method of evolution). The dashed line specifies the phase transition point $J_c=0.329$. 
       (b) The estimated spectral gap in the ferromagnetic phase compared with the series expansion \cite[Table 1]{Oitmaa_1991} up to $\mathcal{O}(g^{20})$. The series diverge at $g \approx 1.5$, and at larger $g$ we show only the iPEPS results.
       }
\end{figure}
\begin{figure}
    \includegraphics[width= \linewidth]{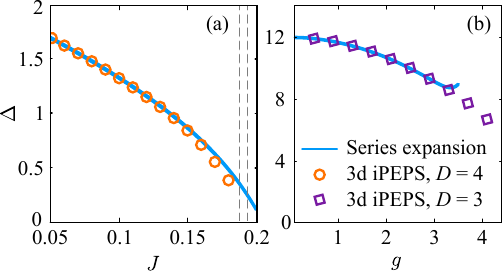}    
    \caption{\label{fig:Gap_3d}%
        The spectral gap estimation for the  3D transverse-field Ising model: (a) in the paramagnetic phase with $D=4$ (circles) and (b) in the ferromagnetic phase with $D=3$ (squares) compared with the series expansion: up to $\mathcal{O}(J^{13})$~\cite[Table 2]{weihong_series_1994} and up to $\mathcal{O}(g^{20})$~\cite[Table 3]{weihong_series_1994}, respectively. The vertical lines indicate the estimates of the critical value $J_c=0.188$~\cite{2013_latorre, braiorr2016phase} and $J_c=0.194$~\cite{weihong_series_1994}.
        } 
\end{figure}

Next, we present the results for the 3D quantum transverse Ising model. We emphasize that variational calculations with iPEPS are not currently available in 3D, while MPS approaches are ineffective due to the rapid growth of the entanglement entropy in 3D systems. The corresponding 3D tensor-network calculations rely primarily on the simple update optimization, accompanied by various methods for subsequent calculation of ground-state observables~\cite{2019_Jahromi, 2020_Jahromi, 2021_Jahromi, 2021_Vlaar}. The 3D Ising model has been previously studied using tensor network methods~\cite{2013_latorre, braiorr2016phase}, making the simple update strategy a natural choice for the gap estimation in these systems. 

In Fig.~\ref{fig:Gap_3d}(a), we compare the 3D iPEPS estimates of the spectral gap with the series expansion from Ref.~\cite{weihong_series_1994} in the paramagnetic regime. The critical parameter $J_{c} = 0.194$ is suggested by the series expansions~\cite{weihong_series_1994}, with an approximate $J_{c} = 0.188$ from previous iPEPS calculations~\cite{2013_latorre, braiorr2016phase}. Our results for the gap value agree with these predictions, as we observe a slowdown in the exponential decay in the regime $J > 0.18$, where the gap estimation becomes challenging. Deep in the paramagnetic phase, our results are consistent with the series expansion predictions. However, in the critical regime, assessing the accuracy of our results is difficult, since the series expansions overestimate the gap, predicting $\Delta_{SE}>0$ at $J=0.2$ (see Fig.~\ref{fig:Gap_3d}), which exceeds the critical point estimated by other series expansions. In Fig.~\ref{fig:Gap_3d}(b) we also show the results in the ferromagnetic phase. As in the 2D case, the series expansions diverge long before the critical point. Hence, we compare our results with the series expansions only for low $g$, where we find agreement between our results and the series expansions. 

So far, our discussion has focused on Ising systems in both symmetry-broken and paramagnetic phases. To verify the applicability of our proposed method to different types of systems, particularly those with symmetry-protected topological order, we have extended our analysis to calculate the gap for the 1D Haldane model, as introduced in Sec.~\ref{subsec:2c}. Our estimated numerical value of the gap, $\Delta = 0.410$, is in close agreement with previously published numerical results, which reported $\Delta = 0.410479(1)$ \cite{Nakano_2009}. The convergence of the numerical derivative of the commutator expectation value is illustrated in Fig.~\ref{fig_5:Haldane}. It is also worth noting that, for 1D systems, the simple update method for calculating observables is exact.
\begin{figure}
    \includegraphics[width= \linewidth]{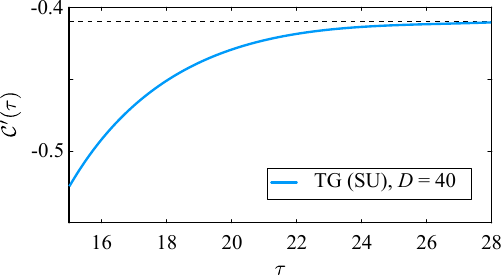}    
    \caption{\label{fig_5:Haldane}%
        The numerical derivative of the expectation value ${\cal C}(\tau)=\ln|\langle [H, O] \rangle(\tau)|$ with the imaginary time $\tau$  for the Haldane model. The operator $O = \sum_{i} S_{i}^{y} S_{i+1}^{z}$.        
        } 
\end{figure}

\section{Conclusions}\label{Sec:Conclusion}

We have presented a tensor network realization of spectral gap calculations via Eq.~\eqref{eq:first_spectral_gap}. Its implementation is straightforward, and only required a few extra lines to be added to an existing code for the simple update method. The runtime of our method scales as $\mathcal{O}(D^{z+1})$, where $z$ is the lattice connectivity (i.e., $z=4,6$ for our 2D and 3D illustrations, respectively). Concretely, it takes about 10 min with MPO on a typical desktop computer to estimate one spectral gap in 3D [pointmarker corresponding to the bond dimension $D=3$ in Fig.~\ref{fig:Gap_3d}(b)]. Our method is complementary to the more powerful but computationally very demanding
variational iPEPS approach~\cite{2020_Ponsioen}, which scales as $\mathcal{O}(D^{10}) - \mathcal{O}(D^{12})$; moreover, the latter method is only applicable to infinite regular 2D lattices and, unlike our approach, cannot be currently extended to 3D.

Looking ahead, there is significant potential for further improvement of the methodology presented. For instance, it can be generalized to arbitrary graphs with low connectivity~\cite{2023_Tindall,2023_tindall_IBM} and to arbitrary unit cells without requiring any symmetries in the system. The use of modified environments, such as cluster environments or nearest-neighbor updates, may lead to more accurate calculations, at the expense of increased computational costs. These improvements would result in more precise calculations, particularly in critical regions. This work is part of a broader program aimed at using the simple update method to study gapped systems, which, in particular, was recently applied to simulate the IBM kicked Ising experiment~\cite{2023_Patra_IBM, 2023_tindall_IBM}.

Finally, we add that although the present work has focused on the spectral gap as a fundamental property in quantum and material sciences, it also plays an important role in vibration analysis, graph theory and network analysis, and data science. Inspired by the interdisciplinary success of tensor networks~\cite{klus2018tensor, gelss2018tensor, klus2019tensor, lucke2022tgedmd}, adaption of our method to these other domains should be a subject of subsequent investigations.

\

\acknowledgments
The authors acknowledge helpful discussions with Andrew Baczewski. I.V.L. and A.G.S. are supported by Office of Naval Research Global (ONRG) Grant No.~N62909-23-1-2088 (program manager Dr.~ Martina Barnas). D.I.B. is supported by Army Research Office (ARO) Grant No.~W911NF-23-1-0288 (program manager Dr.~James Joseph). D.I.B and A.G.S. are also grateful to the Carol Lavin Bernick Faculty Grant Program at Tulane University and the National Science Foundation (NSF) IMPRESS-U Grant No.~2403609. J.M.L. and A.B.M. are supported by Sandia National Laboratories’ Laboratory Directed Research and Development Program under the Truman Fellowship. The views and conclusions contained in this document are those of the authors and should not be interpreted as representing the official policies, either expressed or implied, of NSF, ARO, ONRG or the U.S. Government. The U.S. Government is authorized to reproduce and distribute reprints for Government purposes notwithstanding any copyright notation herein.
This article has been authored by an employee of National Technology \& Engineering Solutions of Sandia, LLC under Contract No. DE-NA0003525 with the U.S. Department of Energy (DOE). The employee owns all right, title and interest in and to the article and is solely responsible for its contents. The United States Government retains and the publisher, by accepting the article for publication, acknowledges that the United States Government retains a nonexclusive, paid-up, irrevocable, world-wide license to publish or reproduce the published form of this article or allow others to do so, for United States Government purposes. The DOE will provide public access to these results of federally sponsored research in accordance with the DOE Public Access Plan~\footnote{\href{https://www.energy.gov/downloads/doe-public-access-plan}{https://www.energy.gov/downloads/doe-public-access-plan}}. This paper describes objective technical results and analysis. Any subjective views or opinions that might be expressed in the paper do not necessarily represent the views of the U.S. Department of Energy or the United States Government.

\bibliography{references}

\begin{thebibliography}{92}%
\makeatletter
\providecommand \@ifxundefined [1]{%
 \@ifx{#1\undefined}
}%
\providecommand \@ifnum [1]{%
 \ifnum #1\expandafter \@firstoftwo
 \else \expandafter \@secondoftwo
 \fi
}%
\providecommand \@ifx [1]{%
 \ifx #1\expandafter \@firstoftwo
 \else \expandafter \@secondoftwo
 \fi
}%
\providecommand \natexlab [1]{#1}%
\providecommand \enquote  [1]{``#1''}%
\providecommand \bibnamefont  [1]{#1}%
\providecommand \bibfnamefont [1]{#1}%
\providecommand \citenamefont [1]{#1}%
\providecommand \href@noop [0]{\@secondoftwo}%
\providecommand \href [0]{\begingroup \@sanitize@url \@href}%
\providecommand \@href[1]{\@@startlink{#1}\@@href}%
\providecommand \@@href[1]{\endgroup#1\@@endlink}%
\providecommand \@sanitize@url [0]{\catcode `\\12\catcode `\$12\catcode `\&12\catcode `\#12\catcode `\^12\catcode `\_12\catcode `\%12\relax}%
\providecommand \@@startlink[1]{}%
\providecommand \@@endlink[0]{}%
\providecommand \url  [0]{\begingroup\@sanitize@url \@url }%
\providecommand \@url [1]{\endgroup\@href {#1}{\urlprefix }}%
\providecommand \urlprefix  [0]{URL }%
\providecommand \Eprint [0]{\href }%
\providecommand \doibase [0]{https://doi.org/}%
\providecommand \selectlanguage [0]{\@gobble}%
\providecommand \bibinfo  [0]{\@secondoftwo}%
\providecommand \bibfield  [0]{\@secondoftwo}%
\providecommand \translation [1]{[#1]}%
\providecommand \BibitemOpen [0]{}%
\providecommand \bibitemStop [0]{}%
\providecommand \bibitemNoStop [0]{.\EOS\space}%
\providecommand \EOS [0]{\spacefactor3000\relax}%
\providecommand \BibitemShut  [1]{\csname bibitem#1\endcsname}%
\let\auto@bib@innerbib\@empty
\bibitem [{\citenamefont {Han}\ \emph {et~al.}(2012)\citenamefont {Han}, \citenamefont {Helton}, \citenamefont {Chu}, \citenamefont {Nocera}, \citenamefont {Rodriguez-Rivera}, \citenamefont {Broholm},\ and\ \citenamefont {Lee}}]{han_fractionalized_2012}%
  \BibitemOpen
  \bibfield  {author} {\bibinfo {author} {\bibfnamefont {T.-H.}\ \bibnamefont {Han}}, \bibinfo {author} {\bibfnamefont {J.~S.}\ \bibnamefont {Helton}}, \bibinfo {author} {\bibfnamefont {S.}~\bibnamefont {Chu}}, \bibinfo {author} {\bibfnamefont {D.~G.}\ \bibnamefont {Nocera}}, \bibinfo {author} {\bibfnamefont {J.~A.}\ \bibnamefont {Rodriguez-Rivera}}, \bibinfo {author} {\bibfnamefont {C.}~\bibnamefont {Broholm}},\ and\ \bibinfo {author} {\bibfnamefont {Y.~S.}\ \bibnamefont {Lee}},\ }\bibfield  {title} {\bibinfo {title} {Fractionalized excitations in the spin-liquid state of a kagome-lattice antiferromagnet},\ }\href {https://doi.org/10.1038/nature11659} {\bibfield  {journal} {\bibinfo  {journal} {Nature}\ }\textbf {\bibinfo {volume} {492}},\ \bibinfo {pages} {406} (\bibinfo {year} {2012})}\BibitemShut {NoStop}%
\bibitem [{\citenamefont {Albash}\ and\ \citenamefont {Lidar}(2018)}]{albash_adiabatic_2018}%
  \BibitemOpen
  \bibfield  {author} {\bibinfo {author} {\bibfnamefont {T.}~\bibnamefont {Albash}}\ and\ \bibinfo {author} {\bibfnamefont {D.~A.}\ \bibnamefont {Lidar}},\ }\bibfield  {title} {\bibinfo {title} {Adiabatic quantum computation},\ }\href {https://doi.org/10.1103/RevModPhys.90.015002} {\bibfield  {journal} {\bibinfo  {journal} {Rev. Mod. Phys.}\ }\textbf {\bibinfo {volume} {90}},\ \bibinfo {pages} {015002} (\bibinfo {year} {2018})}\BibitemShut {NoStop}%
\bibitem [{\citenamefont {Weihong}\ \emph {et~al.}(1994)\citenamefont {Weihong}, \citenamefont {Oitmaa},\ and\ \citenamefont {Hamer}}]{weihong_series_1994}%
  \BibitemOpen
  \bibfield  {author} {\bibinfo {author} {\bibfnamefont {Z.}~\bibnamefont {Weihong}}, \bibinfo {author} {\bibfnamefont {J.}~\bibnamefont {Oitmaa}},\ and\ \bibinfo {author} {\bibfnamefont {C.~J.}\ \bibnamefont {Hamer}},\ }\bibfield  {title} {\bibinfo {title} {Series expansions for the {3D} transverse {Ising} model at ${T}=0$},\ }\href {https://doi.org/10.1088/0305-4470/27/16/010} {\bibfield  {journal} {\bibinfo  {journal} {J. Phys. A}\ }\textbf {\bibinfo {volume} {27}},\ \bibinfo {pages} {5425} (\bibinfo {year} {1994})}\BibitemShut {NoStop}%
\bibitem [{\citenamefont {Pishtshev}(2007)}]{pishtshev_new_2007}%
  \BibitemOpen
  \bibfield  {author} {\bibinfo {author} {\bibfnamefont {A.}~\bibnamefont {Pishtshev}},\ }\bibfield  {title} {\bibinfo {title} {New relation for critical exponents in the {Ising} model},\ }\href {https://doi.org/10.1016/j.physleta.2006.09.008} {\bibfield  {journal} {\bibinfo  {journal} {Phys. Lett. A}\ }\textbf {\bibinfo {volume} {361}},\ \bibinfo {pages} {152} (\bibinfo {year} {2007})}\BibitemShut {NoStop}%
\bibitem [{\citenamefont {Leamer}\ \emph {et~al.}(2023)\citenamefont {Leamer}, \citenamefont {Magann}, \citenamefont {Baczewski}, \citenamefont {McCaul},\ and\ \citenamefont {Bondar}}]{leamer2023spectral}%
  \BibitemOpen
  \bibfield  {author} {\bibinfo {author} {\bibfnamefont {J.~M.}\ \bibnamefont {Leamer}}, \bibinfo {author} {\bibfnamefont {A.~B.}\ \bibnamefont {Magann}}, \bibinfo {author} {\bibfnamefont {A.~D.}\ \bibnamefont {Baczewski}}, \bibinfo {author} {\bibfnamefont {G.}~\bibnamefont {McCaul}},\ and\ \bibinfo {author} {\bibfnamefont {D.~I.}\ \bibnamefont {Bondar}},\ }\href@noop {} {\bibinfo {title} {Spectral gaps via imaginary time}} (\bibinfo {year} {2023}),\ \Eprint {https://arxiv.org/abs/2303.02124} {arXiv:2303.02124 [quant-ph]} \BibitemShut {NoStop}%
\bibitem [{\citenamefont {Koh}(2016)}]{koh_effects_2016}%
  \BibitemOpen
  \bibfield  {author} {\bibinfo {author} {\bibfnamefont {Y.~W.}\ \bibnamefont {Koh}},\ }\bibfield  {title} {\bibinfo {title} {Effects of low-lying excitations on ground-state energy and energy gap of the {Sherrington}-{Kirkpatrick} model in a transverse field},\ }\href {https://doi.org/10.1103/PhysRevB.93.134202} {\bibfield  {journal} {\bibinfo  {journal} {Phys. Rev. B}\ }\textbf {\bibinfo {volume} {93}},\ \bibinfo {pages} {134202} (\bibinfo {year} {2016})}\BibitemShut {NoStop}%
\bibitem [{\citenamefont {Hlatshwayo}\ \emph {et~al.}(2024)\citenamefont {Hlatshwayo}, \citenamefont {Novak},\ and\ \citenamefont {Litvinova}}]{hlatshwayo2023quantum}%
  \BibitemOpen
  \bibfield  {author} {\bibinfo {author} {\bibfnamefont {M.~Q.}\ \bibnamefont {Hlatshwayo}}, \bibinfo {author} {\bibfnamefont {J.}~\bibnamefont {Novak}},\ and\ \bibinfo {author} {\bibfnamefont {E.}~\bibnamefont {Litvinova}},\ }\bibfield  {title} {\bibinfo {title} {Quantum benefit of the quantum equation of motion for the strongly coupled many-body problem},\ }\href {https://doi.org/10.1103/PhysRevC.109.014306} {\bibfield  {journal} {\bibinfo  {journal} {Phys. Rev. C}\ }\textbf {\bibinfo {volume} {109}},\ \bibinfo {pages} {014306} (\bibinfo {year} {2024})}\BibitemShut {NoStop}%
\bibitem [{\citenamefont {Ceperley}\ and\ \citenamefont {Bernu}(1988)}]{ceperley_calculation_1988}%
  \BibitemOpen
  \bibfield  {author} {\bibinfo {author} {\bibfnamefont {D.~M.}\ \bibnamefont {Ceperley}}\ and\ \bibinfo {author} {\bibfnamefont {B.}~\bibnamefont {Bernu}},\ }\bibfield  {title} {\bibinfo {title} {The calculation of excited state properties with quantum {Monte} {Carlo}},\ }\href {https://doi.org/10.1063/1.455398} {\bibfield  {journal} {\bibinfo  {journal} {J. Chem. Phys.}\ }\textbf {\bibinfo {volume} {89}},\ \bibinfo {pages} {6316} (\bibinfo {year} {1988})}\BibitemShut {NoStop}%
\bibitem [{\citenamefont {Caffarel}\ and\ \citenamefont {Claverie}(1988)}]{caffarel_development_1988}%
  \BibitemOpen
  \bibfield  {author} {\bibinfo {author} {\bibfnamefont {M.}~\bibnamefont {Caffarel}}\ and\ \bibinfo {author} {\bibfnamefont {P.}~\bibnamefont {Claverie}},\ }\bibfield  {title} {\bibinfo {title} {Development of a pure diffusion quantum {Monte} {Carlo} method using a full generalized {Feynman}–{Kac} formula. {I}. {Formalism}},\ }\href {https://doi.org/10.1063/1.454227} {\bibfield  {journal} {\bibinfo  {journal} {J. Chem. Phys.}\ }\textbf {\bibinfo {volume} {88}},\ \bibinfo {pages} {1088} (\bibinfo {year} {1988})}\BibitemShut {NoStop}%
\bibitem [{\citenamefont {Foulkes}\ \emph {et~al.}(2001)\citenamefont {Foulkes}, \citenamefont {Mitas}, \citenamefont {Needs},\ and\ \citenamefont {Rajagopal}}]{RevModPhys.73.33}%
  \BibitemOpen
  \bibfield  {author} {\bibinfo {author} {\bibfnamefont {W.~M.~C.}\ \bibnamefont {Foulkes}}, \bibinfo {author} {\bibfnamefont {L.}~\bibnamefont {Mitas}}, \bibinfo {author} {\bibfnamefont {R.~J.}\ \bibnamefont {Needs}},\ and\ \bibinfo {author} {\bibfnamefont {G.}~\bibnamefont {Rajagopal}},\ }\bibfield  {title} {\bibinfo {title} {Quantum monte carlo simulations of solids},\ }\href {https://doi.org/10.1103/RevModPhys.73.33} {\bibfield  {journal} {\bibinfo  {journal} {Rev. Mod. Phys.}\ }\textbf {\bibinfo {volume} {73}},\ \bibinfo {pages} {33} (\bibinfo {year} {2001})}\BibitemShut {NoStop}%
\bibitem [{\citenamefont {Motta}\ \emph {et~al.}(2020)\citenamefont {Motta}, \citenamefont {Sun}, \citenamefont {Tan}, \citenamefont {O’Rourke}, \citenamefont {Ye}, \citenamefont {Minnich}, \citenamefont {Brandao},\ and\ \citenamefont {Chan}}]{motta2020determining}%
  \BibitemOpen
  \bibfield  {author} {\bibinfo {author} {\bibfnamefont {M.}~\bibnamefont {Motta}}, \bibinfo {author} {\bibfnamefont {C.}~\bibnamefont {Sun}}, \bibinfo {author} {\bibfnamefont {A.~T.}\ \bibnamefont {Tan}}, \bibinfo {author} {\bibfnamefont {M.~J.}\ \bibnamefont {O’Rourke}}, \bibinfo {author} {\bibfnamefont {E.}~\bibnamefont {Ye}}, \bibinfo {author} {\bibfnamefont {A.~J.}\ \bibnamefont {Minnich}}, \bibinfo {author} {\bibfnamefont {F.~G.}\ \bibnamefont {Brandao}},\ and\ \bibinfo {author} {\bibfnamefont {G.~K.-L.}\ \bibnamefont {Chan}},\ }\bibfield  {title} {\bibinfo {title} {Determining eigenstates and thermal states on a quantum computer using quantum imaginary time evolution},\ }\href {https://doi.org/10.1038/s41567-019-0704-4} {\bibfield  {journal} {\bibinfo  {journal} {Nat. Phys.}\ }\textbf {\bibinfo {volume} {16}},\ \bibinfo {pages} {205} (\bibinfo {year} {2020})}\BibitemShut {NoStop}%
\bibitem [{\citenamefont {Russo}\ \emph {et~al.}(2021)\citenamefont {Russo}, \citenamefont {Rudinger}, \citenamefont {Morrison},\ and\ \citenamefont {Baczewski}}]{russo_evaluating_2021}%
  \BibitemOpen
  \bibfield  {author} {\bibinfo {author} {\bibfnamefont {A.~E.}\ \bibnamefont {Russo}}, \bibinfo {author} {\bibfnamefont {K.~M.}\ \bibnamefont {Rudinger}}, \bibinfo {author} {\bibfnamefont {B.~C.~A.}\ \bibnamefont {Morrison}},\ and\ \bibinfo {author} {\bibfnamefont {A.~D.}\ \bibnamefont {Baczewski}},\ }\bibfield  {title} {\bibinfo {title} {Evaluating {Energy} {Differences} on a {Quantum} {Computer} with {Robust} {Phase} {Estimation}},\ }\href {https://doi.org/10.1103/PhysRevLett.126.210501} {\bibfield  {journal} {\bibinfo  {journal} {Phys. Rev. Lett.}\ }\textbf {\bibinfo {volume} {126}},\ \bibinfo {pages} {210501} (\bibinfo {year} {2021})}\BibitemShut {NoStop}%
\bibitem [{\citenamefont {Gnatenko}\ \emph {et~al.}(2022{\natexlab{a}})\citenamefont {Gnatenko}, \citenamefont {Laba},\ and\ \citenamefont {Tkachuk}}]{gnatenko_detection_2022}%
  \BibitemOpen
  \bibfield  {author} {\bibinfo {author} {\bibfnamefont {K.~P.}\ \bibnamefont {Gnatenko}}, \bibinfo {author} {\bibfnamefont {H.~P.}\ \bibnamefont {Laba}},\ and\ \bibinfo {author} {\bibfnamefont {V.~M.}\ \bibnamefont {Tkachuk}},\ }\bibfield  {title} {\bibinfo {title} {Detection of energy levels of a spin system on a quantum computer by probe spin evolution},\ }\href {https://doi.org/10.1140/epjp/s13360-022-02753-0} {\bibfield  {journal} {\bibinfo  {journal} {Eur. Phys. J. Plus}\ }\textbf {\bibinfo {volume} {137}},\ \bibinfo {pages} {522} (\bibinfo {year} {2022}{\natexlab{a}})}\BibitemShut {NoStop}%
\bibitem [{\citenamefont {Gnatenko}\ \emph {et~al.}(2022{\natexlab{b}})\citenamefont {Gnatenko}, \citenamefont {Laba},\ and\ \citenamefont {Tkachuk}}]{gnatenko_energy_2022}%
  \BibitemOpen
  \bibfield  {author} {\bibinfo {author} {\bibfnamefont {K.~P.}\ \bibnamefont {Gnatenko}}, \bibinfo {author} {\bibfnamefont {H.}~\bibnamefont {Laba}},\ and\ \bibinfo {author} {\bibfnamefont {V.}~\bibnamefont {Tkachuk}},\ }\bibfield  {title} {\bibinfo {title} {Energy levels estimation on a quantum computer by evolution of a physical quantity},\ }\href {https://doi.org/10.1016/j.physleta.2021.127843} {\bibfield  {journal} {\bibinfo  {journal} {Phys. Lett. A}\ }\textbf {\bibinfo {volume} {424}},\ \bibinfo {pages} {127843} (\bibinfo {year} {2022}{\natexlab{b}})}\BibitemShut {NoStop}%
\bibitem [{\citenamefont {Stroeks}\ \emph {et~al.}(2022)\citenamefont {Stroeks}, \citenamefont {Helsen},\ and\ \citenamefont {Terhal}}]{stroeks_spectral_2022}%
  \BibitemOpen
  \bibfield  {author} {\bibinfo {author} {\bibfnamefont {M.~E.}\ \bibnamefont {Stroeks}}, \bibinfo {author} {\bibfnamefont {J.}~\bibnamefont {Helsen}},\ and\ \bibinfo {author} {\bibfnamefont {B.~M.}\ \bibnamefont {Terhal}},\ }\bibfield  {title} {\bibinfo {title} {Spectral estimation for {Hamiltonians}: a comparison between classical imaginary-time evolution and quantum real-time evolution},\ }\href {https://doi.org/10.1088/1367-2630/ac919c} {\bibfield  {journal} {\bibinfo  {journal} {New J. Phys.}\ }\textbf {\bibinfo {volume} {24}},\ \bibinfo {pages} {103024} (\bibinfo {year} {2022})}\BibitemShut {NoStop}%
\bibitem [{\citenamefont {Leamer}\ \emph {et~al.}(2024)\citenamefont {Leamer}, \citenamefont {Bondar},\ and\ \citenamefont {McCaul}}]{leamer2024quantum}%
  \BibitemOpen
  \bibfield  {author} {\bibinfo {author} {\bibfnamefont {J.~M.}\ \bibnamefont {Leamer}}, \bibinfo {author} {\bibfnamefont {D.~I.}\ \bibnamefont {Bondar}},\ and\ \bibinfo {author} {\bibfnamefont {G.}~\bibnamefont {McCaul}},\ }\bibfield  {title} {\bibinfo {title} {Quantum dynamical emulation},\ }\bibfield  {journal} {\bibinfo  {journal} {arXiv:2403.03350}\ }\href {https://doi.org/10.48550/arXiv.2403.03350} {10.48550/arXiv.2403.03350} (\bibinfo {year} {2024})\BibitemShut {NoStop}%
\bibitem [{\citenamefont {Or{\'u}s}(2019)}]{Orus2019review}%
  \BibitemOpen
  \bibfield  {author} {\bibinfo {author} {\bibfnamefont {R.}~\bibnamefont {Or{\'u}s}},\ }\bibfield  {title} {\bibinfo {title} {Tensor networks for complex quantum systems},\ }\href {https://doi.org/10.1038/s42254-019-0086-7} {\bibfield  {journal} {\bibinfo  {journal} {Nat. Rev. Phys.}\ }\textbf {\bibinfo {volume} {1}},\ \bibinfo {pages} {538} (\bibinfo {year} {2019})}\BibitemShut {NoStop}%
\bibitem [{\citenamefont {Cirac}\ \emph {et~al.}(2021)\citenamefont {Cirac}, \citenamefont {P\'erez-Garc\'{\i}a}, \citenamefont {Schuch},\ and\ \citenamefont {Verstraete}}]{2021_Cirac}%
  \BibitemOpen
  \bibfield  {author} {\bibinfo {author} {\bibfnamefont {J.~I.}\ \bibnamefont {Cirac}}, \bibinfo {author} {\bibfnamefont {D.}~\bibnamefont {P\'erez-Garc\'{\i}a}}, \bibinfo {author} {\bibfnamefont {N.}~\bibnamefont {Schuch}},\ and\ \bibinfo {author} {\bibfnamefont {F.}~\bibnamefont {Verstraete}},\ }\bibfield  {title} {\bibinfo {title} {Matrix product states and projected entangled pair states: Concepts, symmetries, theorems},\ }\href {https://doi.org/10.1103/RevModPhys.93.045003} {\bibfield  {journal} {\bibinfo  {journal} {Rev. Mod. Phys.}\ }\textbf {\bibinfo {volume} {93}},\ \bibinfo {pages} {045003} (\bibinfo {year} {2021})}\BibitemShut {NoStop}%
\bibitem [{\citenamefont {Ba\~{n}uls}(2023)}]{2023_Banuls}%
  \BibitemOpen
  \bibfield  {author} {\bibinfo {author} {\bibfnamefont {M.~C.}\ \bibnamefont {Ba\~{n}uls}},\ }\bibfield  {title} {\bibinfo {title} {Tensor network algorithms: A route map},\ }\href {https://doi.org/10.1146/annurev-conmatphys-040721-022705} {\bibfield  {journal} {\bibinfo  {journal} {Annu. Rev. Cond. Matt. Phys.}\ }\textbf {\bibinfo {volume} {14}},\ \bibinfo {pages} {173} (\bibinfo {year} {2023})}\BibitemShut {NoStop}%
\bibitem [{\citenamefont {Okunishi}\ \emph {et~al.}(2022)\citenamefont {Okunishi}, \citenamefont {Nishino},\ and\ \citenamefont {Ueda}}]{Okunishi_2022}%
  \BibitemOpen
  \bibfield  {author} {\bibinfo {author} {\bibfnamefont {K.}~\bibnamefont {Okunishi}}, \bibinfo {author} {\bibfnamefont {T.}~\bibnamefont {Nishino}},\ and\ \bibinfo {author} {\bibfnamefont {H.}~\bibnamefont {Ueda}},\ }\bibfield  {title} {\bibinfo {title} {Developments in the tensor network — from statistical mechanics to quantum entanglement},\ }\href {https://doi.org/10.7566/JPSJ.91.062001} {\bibfield  {journal} {\bibinfo  {journal} {J. Phys. Soc. Jpn.}\ }\textbf {\bibinfo {volume} {91}},\ \bibinfo {pages} {062001} (\bibinfo {year} {2022})}\BibitemShut {NoStop}%
\bibitem [{\citenamefont {Bruognolo}\ \emph {et~al.}(2021)\citenamefont {Bruognolo}, \citenamefont {Li}, \citenamefont {von Delft},\ and\ \citenamefont {Weichselbaum}}]{2021_Bruognolo}%
  \BibitemOpen
  \bibfield  {author} {\bibinfo {author} {\bibfnamefont {B.}~\bibnamefont {Bruognolo}}, \bibinfo {author} {\bibfnamefont {J.-W.}\ \bibnamefont {Li}}, \bibinfo {author} {\bibfnamefont {J.}~\bibnamefont {von Delft}},\ and\ \bibinfo {author} {\bibfnamefont {A.}~\bibnamefont {Weichselbaum}},\ }\bibfield  {title} {\bibinfo {title} {{A beginner's guide to non-abelian iPEPS for correlated fermions}},\ }\href {https://doi.org/10.21468/SciPostPhysLectNotes.25} {\bibfield  {journal} {\bibinfo  {journal} {SciPost Phys. Lect. Notes}\ }\textbf {\bibinfo {volume} {25}},\ \bibinfo {pages} {1} (\bibinfo {year} {2021})}\BibitemShut {NoStop}%
\bibitem [{\citenamefont {Vidal}(2004)}]{2004_Vidal}%
  \BibitemOpen
  \bibfield  {author} {\bibinfo {author} {\bibfnamefont {G.}~\bibnamefont {Vidal}},\ }\bibfield  {title} {\bibinfo {title} {Efficient simulation of one-dimensional quantum many-body systems},\ }\href {https://doi.org/10.1103/PhysRevLett.93.040502} {\bibfield  {journal} {\bibinfo  {journal} {Phys. Rev. Lett.}\ }\textbf {\bibinfo {volume} {93}},\ \bibinfo {pages} {040502} (\bibinfo {year} {2004})}\BibitemShut {NoStop}%
\bibitem [{\citenamefont {Schollwöck}(2011)}]{SCHOLLWOCK201196}%
  \BibitemOpen
  \bibfield  {author} {\bibinfo {author} {\bibfnamefont {U.}~\bibnamefont {Schollwöck}},\ }\bibfield  {title} {\bibinfo {title} {The density-matrix renormalization group in the age of matrix product states},\ }\href {https://doi.org/https://doi.org/10.1016/j.aop.2010.09.012} {\bibfield  {journal} {\bibinfo  {journal} {Ann. Phys.}\ }\textbf {\bibinfo {volume} {326}},\ \bibinfo {pages} {96} (\bibinfo {year} {2011})}\BibitemShut {NoStop}%
\bibitem [{\citenamefont {White}(1992)}]{1992_White}%
  \BibitemOpen
  \bibfield  {author} {\bibinfo {author} {\bibfnamefont {S.~R.}\ \bibnamefont {White}},\ }\bibfield  {title} {\bibinfo {title} {Density matrix formulation for quantum renormalization groups},\ }\href {https://doi.org/10.1103/PhysRevLett.69.2863} {\bibfield  {journal} {\bibinfo  {journal} {Phys. Rev. Lett.}\ }\textbf {\bibinfo {volume} {69}},\ \bibinfo {pages} {2863} (\bibinfo {year} {1992})}\BibitemShut {NoStop}%
\bibitem [{\citenamefont {F.~Verstraete}\ and\ \citenamefont {Cirac}(2008)}]{2008_Verstraete}%
  \BibitemOpen
  \bibfield  {author} {\bibinfo {author} {\bibfnamefont {V.~M.}\ \bibnamefont {F.~Verstraete}}\ and\ \bibinfo {author} {\bibfnamefont {J.}~\bibnamefont {Cirac}},\ }\bibfield  {title} {\bibinfo {title} {Matrix product states, projected entangled pair states, and variational renormalization group methods for quantum spin systems},\ }\href {https://doi.org/10.1080/14789940801912366} {\bibfield  {journal} {\bibinfo  {journal} {Adv. Phys.}\ }\textbf {\bibinfo {volume} {57}},\ \bibinfo {pages} {143} (\bibinfo {year} {2008})}\BibitemShut {NoStop}%
\bibitem [{\citenamefont {Corboz}\ \emph {et~al.}(2012)\citenamefont {Corboz}, \citenamefont {Lajk\'o}, \citenamefont {L\"auchli}, \citenamefont {Penc},\ and\ \citenamefont {Mila}}]{2012_Corboz}%
  \BibitemOpen
  \bibfield  {author} {\bibinfo {author} {\bibfnamefont {P.}~\bibnamefont {Corboz}}, \bibinfo {author} {\bibfnamefont {M.}~\bibnamefont {Lajk\'o}}, \bibinfo {author} {\bibfnamefont {A.~M.}\ \bibnamefont {L\"auchli}}, \bibinfo {author} {\bibfnamefont {K.}~\bibnamefont {Penc}},\ and\ \bibinfo {author} {\bibfnamefont {F.}~\bibnamefont {Mila}},\ }\bibfield  {title} {\bibinfo {title} {Spin-orbital quantum liquid on the honeycomb lattice},\ }\href {https://doi.org/10.1103/PhysRevX.2.041013} {\bibfield  {journal} {\bibinfo  {journal} {Phys. Rev. X}\ }\textbf {\bibinfo {volume} {2}},\ \bibinfo {pages} {041013} (\bibinfo {year} {2012})}\BibitemShut {NoStop}%
\bibitem [{\citenamefont {Corboz}\ and\ \citenamefont {Mila}(2013)}]{2013_Corboz}%
  \BibitemOpen
  \bibfield  {author} {\bibinfo {author} {\bibfnamefont {P.}~\bibnamefont {Corboz}}\ and\ \bibinfo {author} {\bibfnamefont {F.}~\bibnamefont {Mila}},\ }\bibfield  {title} {\bibinfo {title} {Tensor network study of the {Shastry-Sutherland} model in zero magnetic field},\ }\href {https://doi.org/10.1103/PhysRevB.87.115144} {\bibfield  {journal} {\bibinfo  {journal} {Phys. Rev. B}\ }\textbf {\bibinfo {volume} {87}},\ \bibinfo {pages} {115144} (\bibinfo {year} {2013})}\BibitemShut {NoStop}%
\bibitem [{\citenamefont {Corboz}\ \emph {et~al.}(2014)\citenamefont {Corboz}, \citenamefont {Rice},\ and\ \citenamefont {Troyer}}]{2014_Corboz}%
  \BibitemOpen
  \bibfield  {author} {\bibinfo {author} {\bibfnamefont {P.}~\bibnamefont {Corboz}}, \bibinfo {author} {\bibfnamefont {T.~M.}\ \bibnamefont {Rice}},\ and\ \bibinfo {author} {\bibfnamefont {M.}~\bibnamefont {Troyer}},\ }\bibfield  {title} {\bibinfo {title} {Competing states in the {$t$-$J$} model: Uniform $d$-wave state versus stripe state},\ }\href {https://doi.org/10.1103/PhysRevLett.113.046402} {\bibfield  {journal} {\bibinfo  {journal} {Phys. Rev. Lett.}\ }\textbf {\bibinfo {volume} {113}},\ \bibinfo {pages} {046402} (\bibinfo {year} {2014})}\BibitemShut {NoStop}%
\bibitem [{\citenamefont {Corboz}(2016)}]{2015_Corboz}%
  \BibitemOpen
  \bibfield  {author} {\bibinfo {author} {\bibfnamefont {P.}~\bibnamefont {Corboz}},\ }\bibfield  {title} {\bibinfo {title} {Improved energy extrapolation with infinite projected entangled-pair states applied to the two-dimensional {Hubbard} model},\ }\href {https://doi.org/10.1103/PhysRevB.93.045116} {\bibfield  {journal} {\bibinfo  {journal} {Phys. Rev. B}\ }\textbf {\bibinfo {volume} {93}},\ \bibinfo {pages} {045116} (\bibinfo {year} {2016})}\BibitemShut {NoStop}%
\bibitem [{\citenamefont {Hasik}\ \emph {et~al.}(2022)\citenamefont {Hasik}, \citenamefont {Van~Damme}, \citenamefont {Poilblanc},\ and\ \citenamefont {Vanderstraeten}}]{2022_Hasik}%
  \BibitemOpen
  \bibfield  {author} {\bibinfo {author} {\bibfnamefont {J.}~\bibnamefont {Hasik}}, \bibinfo {author} {\bibfnamefont {M.}~\bibnamefont {Van~Damme}}, \bibinfo {author} {\bibfnamefont {D.}~\bibnamefont {Poilblanc}},\ and\ \bibinfo {author} {\bibfnamefont {L.}~\bibnamefont {Vanderstraeten}},\ }\bibfield  {title} {\bibinfo {title} {Simulating chiral spin liquids with projected entangled-pair states},\ }\href {https://doi.org/10.1103/PhysRevLett.129.177201} {\bibfield  {journal} {\bibinfo  {journal} {Phys. Rev. Lett.}\ }\textbf {\bibinfo {volume} {129}},\ \bibinfo {pages} {177201} (\bibinfo {year} {2022})}\BibitemShut {NoStop}%
\bibitem [{\citenamefont {Jahromi}\ and\ \citenamefont {Or{\'u}s}(2020)}]{2020_Jahromi}%
  \BibitemOpen
  \bibfield  {author} {\bibinfo {author} {\bibfnamefont {S.~S.}\ \bibnamefont {Jahromi}}\ and\ \bibinfo {author} {\bibfnamefont {R.}~\bibnamefont {Or{\'u}s}},\ }\bibfield  {title} {\bibinfo {title} {Thermal bosons in 3d optical lattices via tensor networks},\ }\href {https://doi.org/10.1038/s41598-020-75548-x} {\bibfield  {journal} {\bibinfo  {journal} {Sci. Rep.}\ }\textbf {\bibinfo {volume} {10}},\ \bibinfo {pages} {19051} (\bibinfo {year} {2020})}\BibitemShut {NoStop}%
\bibitem [{\citenamefont {Czarnik}\ and\ \citenamefont {Dziarmaga}(2015)}]{2015_Czarnik}%
  \BibitemOpen
  \bibfield  {author} {\bibinfo {author} {\bibfnamefont {P.}~\bibnamefont {Czarnik}}\ and\ \bibinfo {author} {\bibfnamefont {J.}~\bibnamefont {Dziarmaga}},\ }\bibfield  {title} {\bibinfo {title} {Variational approach to projected entangled pair states at finite temperature},\ }\href {https://doi.org/10.1103/PhysRevB.92.035152} {\bibfield  {journal} {\bibinfo  {journal} {Phys. Rev. B}\ }\textbf {\bibinfo {volume} {92}},\ \bibinfo {pages} {035152} (\bibinfo {year} {2015})}\BibitemShut {NoStop}%
\bibitem [{\citenamefont {Czarnik}\ \emph {et~al.}(2019)\citenamefont {Czarnik}, \citenamefont {Dziarmaga},\ and\ \citenamefont {Corboz}}]{2019_Czarnik}%
  \BibitemOpen
  \bibfield  {author} {\bibinfo {author} {\bibfnamefont {P.}~\bibnamefont {Czarnik}}, \bibinfo {author} {\bibfnamefont {J.}~\bibnamefont {Dziarmaga}},\ and\ \bibinfo {author} {\bibfnamefont {P.}~\bibnamefont {Corboz}},\ }\bibfield  {title} {\bibinfo {title} {Time evolution of an infinite projected entangled pair state: An efficient algorithm},\ }\href {https://doi.org/10.1103/PhysRevB.99.035115} {\bibfield  {journal} {\bibinfo  {journal} {Phys. Rev. B}\ }\textbf {\bibinfo {volume} {99}},\ \bibinfo {pages} {035115} (\bibinfo {year} {2019})}\BibitemShut {NoStop}%
\bibitem [{\citenamefont {Vanderstraeten}\ \emph {et~al.}(2019{\natexlab{a}})\citenamefont {Vanderstraeten}, \citenamefont {Haegeman},\ and\ \citenamefont {Verstraete}}]{2019_Vanderstraten_PEPS}%
  \BibitemOpen
  \bibfield  {author} {\bibinfo {author} {\bibfnamefont {L.}~\bibnamefont {Vanderstraeten}}, \bibinfo {author} {\bibfnamefont {J.}~\bibnamefont {Haegeman}},\ and\ \bibinfo {author} {\bibfnamefont {F.}~\bibnamefont {Verstraete}},\ }\bibfield  {title} {\bibinfo {title} {Simulating excitation spectra with projected entangled-pair states},\ }\href {https://doi.org/10.1103/PhysRevB.99.165121} {\bibfield  {journal} {\bibinfo  {journal} {Phys. Rev. B}\ }\textbf {\bibinfo {volume} {99}},\ \bibinfo {pages} {165121} (\bibinfo {year} {2019}{\natexlab{a}})}\BibitemShut {NoStop}%
\bibitem [{\citenamefont {Jiang}\ \emph {et~al.}(2008)\citenamefont {Jiang}, \citenamefont {Weng},\ and\ \citenamefont {Xiang}}]{2008_Jiang}%
  \BibitemOpen
  \bibfield  {author} {\bibinfo {author} {\bibfnamefont {H.~C.}\ \bibnamefont {Jiang}}, \bibinfo {author} {\bibfnamefont {Z.~Y.}\ \bibnamefont {Weng}},\ and\ \bibinfo {author} {\bibfnamefont {T.}~\bibnamefont {Xiang}},\ }\bibfield  {title} {\bibinfo {title} {Accurate determination of tensor network state of quantum lattice models in two dimensions},\ }\href {https://doi.org/10.1103/PhysRevLett.101.090603} {\bibfield  {journal} {\bibinfo  {journal} {Phys. Rev. Lett.}\ }\textbf {\bibinfo {volume} {101}},\ \bibinfo {pages} {090603} (\bibinfo {year} {2008})}\BibitemShut {NoStop}%
\bibitem [{\citenamefont {Jahromi}\ and\ \citenamefont {Or\'us}(2019)}]{2019_Jahromi}%
  \BibitemOpen
  \bibfield  {author} {\bibinfo {author} {\bibfnamefont {S.~S.}\ \bibnamefont {Jahromi}}\ and\ \bibinfo {author} {\bibfnamefont {R.}~\bibnamefont {Or\'us}},\ }\bibfield  {title} {\bibinfo {title} {Universal tensor-network algorithm for any infinite lattice},\ }\href {https://doi.org/10.1103/PhysRevB.99.195105} {\bibfield  {journal} {\bibinfo  {journal} {Phys. Rev. B}\ }\textbf {\bibinfo {volume} {99}},\ \bibinfo {pages} {195105} (\bibinfo {year} {2019})}\BibitemShut {NoStop}%
\bibitem [{\citenamefont {Tindall}\ and\ \citenamefont {Fishman}(2023)}]{2023_Tindall}%
  \BibitemOpen
  \bibfield  {author} {\bibinfo {author} {\bibfnamefont {J.}~\bibnamefont {Tindall}}\ and\ \bibinfo {author} {\bibfnamefont {M.}~\bibnamefont {Fishman}},\ }\bibfield  {title} {\bibinfo {title} {{Gauging tensor networks with belief propagation}},\ }\href {https://doi.org/10.21468/SciPostPhys.15.6.222} {\bibfield  {journal} {\bibinfo  {journal} {SciPost Phys.}\ }\textbf {\bibinfo {volume} {15}},\ \bibinfo {pages} {222} (\bibinfo {year} {2023})}\BibitemShut {NoStop}%
\bibitem [{\citenamefont {Bl\"ote}\ and\ \citenamefont {Deng}(2002)}]{2002_Blote}%
  \BibitemOpen
  \bibfield  {author} {\bibinfo {author} {\bibfnamefont {H.~W.~J.}\ \bibnamefont {Bl\"ote}}\ and\ \bibinfo {author} {\bibfnamefont {Y.}~\bibnamefont {Deng}},\ }\bibfield  {title} {\bibinfo {title} {{Cluster Monte Carlo simulation of the transverse Ising model}},\ }\href {https://doi.org/10.1103/PhysRevE.66.066110} {\bibfield  {journal} {\bibinfo  {journal} {Phys. Rev. E}\ }\textbf {\bibinfo {volume} {66}},\ \bibinfo {pages} {066110} (\bibinfo {year} {2002})}\BibitemShut {NoStop}%
\bibitem [{\citenamefont {Garc\'{\i}a-S\'aez}\ and\ \citenamefont {Latorre}(2013)}]{2013_latorre}%
  \BibitemOpen
  \bibfield  {author} {\bibinfo {author} {\bibfnamefont {A.}~\bibnamefont {Garc\'{\i}a-S\'aez}}\ and\ \bibinfo {author} {\bibfnamefont {J.~I.}\ \bibnamefont {Latorre}},\ }\bibfield  {title} {\bibinfo {title} {Renormalization group contraction of tensor networks in three dimensions},\ }\href {https://doi.org/10.1103/PhysRevB.87.085130} {\bibfield  {journal} {\bibinfo  {journal} {Phys. Rev. B}\ }\textbf {\bibinfo {volume} {87}},\ \bibinfo {pages} {085130} (\bibinfo {year} {2013})}\BibitemShut {NoStop}%
\bibitem [{\citenamefont {Oitmaa}\ \emph {et~al.}(1991)\citenamefont {Oitmaa}, \citenamefont {Hamer},\ and\ \citenamefont {Weihong}}]{Oitmaa_1991}%
  \BibitemOpen
  \bibfield  {author} {\bibinfo {author} {\bibfnamefont {J.}~\bibnamefont {Oitmaa}}, \bibinfo {author} {\bibfnamefont {C.~J.}\ \bibnamefont {Hamer}},\ and\ \bibinfo {author} {\bibfnamefont {Z.}~\bibnamefont {Weihong}},\ }\bibfield  {title} {\bibinfo {title} {Low-temperature series expansions for the (2+1)-dimensional {Ising} model},\ }\href {https://doi.org/10.1088/0305-4470/24/12/024} {\bibfield  {journal} {\bibinfo  {journal} {J. Phys. A: Math. Gen.}\ }\textbf {\bibinfo {volume} {24}},\ \bibinfo {pages} {2863} (\bibinfo {year} {1991})}\BibitemShut {NoStop}%
\bibitem [{\citenamefont {He}\ \emph {et~al.}(1990)\citenamefont {He}, \citenamefont {Hamer},\ and\ \citenamefont {Oitmaa}}]{He_1990}%
  \BibitemOpen
  \bibfield  {author} {\bibinfo {author} {\bibfnamefont {H.~X.}\ \bibnamefont {He}}, \bibinfo {author} {\bibfnamefont {C.~J.}\ \bibnamefont {Hamer}},\ and\ \bibinfo {author} {\bibfnamefont {J.}~\bibnamefont {Oitmaa}},\ }\bibfield  {title} {\bibinfo {title} {High-temperature series expansions for the (2+1)-dimensional {Ising} model},\ }\href {https://doi.org/10.1088/0305-4470/23/10/018} {\bibfield  {journal} {\bibinfo  {journal} {J. Phys. A: Math. Gen.}\ }\textbf {\bibinfo {volume} {23}},\ \bibinfo {pages} {1775} (\bibinfo {year} {1990})}\BibitemShut {NoStop}%
\bibitem [{\citenamefont {Haldane}(1983{\natexlab{a}})}]{HALDANE1983464}%
  \BibitemOpen
  \bibfield  {author} {\bibinfo {author} {\bibfnamefont {F.}~\bibnamefont {Haldane}},\ }\bibfield  {title} {\bibinfo {title} {Continuum dynamics of the {1-D} {Heisenberg} antiferromagnet: Identification with the {O(3)} nonlinear sigma model},\ }\href {https://doi.org/https://doi.org/10.1016/0375-9601(83)90631-X} {\bibfield  {journal} {\bibinfo  {journal} {Phys. Lett. A}\ }\textbf {\bibinfo {volume} {93}},\ \bibinfo {pages} {464} (\bibinfo {year} {1983}{\natexlab{a}})}\BibitemShut {NoStop}%
\bibitem [{\citenamefont {Haldane}(1983{\natexlab{b}})}]{Haldane83}%
  \BibitemOpen
  \bibfield  {author} {\bibinfo {author} {\bibfnamefont {F.~D.~M.}\ \bibnamefont {Haldane}},\ }\bibfield  {title} {\bibinfo {title} {Nonlinear field theory of large-spin {H}eisenberg antiferromagnets: Semiclassically quantized solitons of the one-dimensional easy-axis {N}\'eel state},\ }\href {https://doi.org/10.1103/PhysRevLett.50.1153} {\bibfield  {journal} {\bibinfo  {journal} {Phys. Rev. Lett.}\ }\textbf {\bibinfo {volume} {50}},\ \bibinfo {pages} {1153} (\bibinfo {year} {1983}{\natexlab{b}})}\BibitemShut {NoStop}%
\bibitem [{\citenamefont {Nakano}\ and\ \citenamefont {Terai}(2009)}]{Nakano_2009}%
  \BibitemOpen
  \bibfield  {author} {\bibinfo {author} {\bibfnamefont {H.}~\bibnamefont {Nakano}}\ and\ \bibinfo {author} {\bibfnamefont {A.}~\bibnamefont {Terai}},\ }\bibfield  {title} {\bibinfo {title} {Reexamination of finite-lattice extrapolation of {Haldane} gaps},\ }\href {https://doi.org/10.1143/JPSJ.78.014003} {\bibfield  {journal} {\bibinfo  {journal} {J. Phys. Soc. Jpn.}\ }\textbf {\bibinfo {volume} {78}},\ \bibinfo {pages} {014003} (\bibinfo {year} {2009})}\BibitemShut {NoStop}%
\bibitem [{\citenamefont {Golinelli}\ \emph {et~al.}(1994)\citenamefont {Golinelli}, \citenamefont {Jolicoeur},\ and\ \citenamefont {Lacaze}}]{Golinelli_1994}%
  \BibitemOpen
  \bibfield  {author} {\bibinfo {author} {\bibfnamefont {O.}~\bibnamefont {Golinelli}}, \bibinfo {author} {\bibfnamefont {T.}~\bibnamefont {Jolicoeur}},\ and\ \bibinfo {author} {\bibfnamefont {R.}~\bibnamefont {Lacaze}},\ }\bibfield  {title} {\bibinfo {title} {Finite-lattice extrapolations for a {H}aldane-gap antiferromagnet},\ }\href {https://doi.org/10.1103/PhysRevB.50.3037} {\bibfield  {journal} {\bibinfo  {journal} {Phys. Rev. B}\ }\textbf {\bibinfo {volume} {50}},\ \bibinfo {pages} {3037} (\bibinfo {year} {1994})}\BibitemShut {NoStop}%
\bibitem [{\citenamefont {White}\ and\ \citenamefont {Huse}(1993)}]{White_1993}%
  \BibitemOpen
  \bibfield  {author} {\bibinfo {author} {\bibfnamefont {S.~R.}\ \bibnamefont {White}}\ and\ \bibinfo {author} {\bibfnamefont {D.~A.}\ \bibnamefont {Huse}},\ }\bibfield  {title} {\bibinfo {title} {Numerical renormalization-group study of low-lying eigenstates of the antiferromagnetic ${S}=1$ {H}eisenberg chain},\ }\href {https://doi.org/10.1103/PhysRevB.48.3844} {\bibfield  {journal} {\bibinfo  {journal} {Phys. Rev. B}\ }\textbf {\bibinfo {volume} {48}},\ \bibinfo {pages} {3844} (\bibinfo {year} {1993})}\BibitemShut {NoStop}%
\bibitem [{\citenamefont {Todo}\ and\ \citenamefont {Kato}(2001)}]{Todo_2001}%
  \BibitemOpen
  \bibfield  {author} {\bibinfo {author} {\bibfnamefont {S.}~\bibnamefont {Todo}}\ and\ \bibinfo {author} {\bibfnamefont {K.}~\bibnamefont {Kato}},\ }\bibfield  {title} {\bibinfo {title} {Cluster algorithms for general- $\mathit{S}$ quantum spin systems},\ }\href {https://doi.org/10.1103/PhysRevLett.87.047203} {\bibfield  {journal} {\bibinfo  {journal} {Phys. Rev. Lett.}\ }\textbf {\bibinfo {volume} {87}},\ \bibinfo {pages} {047203} (\bibinfo {year} {2001})}\BibitemShut {NoStop}%
\bibitem [{\citenamefont {Affleck}\ \emph {et~al.}(1987)\citenamefont {Affleck}, \citenamefont {Kennedy}, \citenamefont {Lieb},\ and\ \citenamefont {Tasaki}}]{AKLT}%
  \BibitemOpen
  \bibfield  {author} {\bibinfo {author} {\bibfnamefont {I.}~\bibnamefont {Affleck}}, \bibinfo {author} {\bibfnamefont {T.}~\bibnamefont {Kennedy}}, \bibinfo {author} {\bibfnamefont {E.~H.}\ \bibnamefont {Lieb}},\ and\ \bibinfo {author} {\bibfnamefont {H.}~\bibnamefont {Tasaki}},\ }\bibfield  {title} {\bibinfo {title} {Rigorous results on valence-bond ground states in antiferromagnets},\ }\href {https://doi.org/10.1103/PhysRevLett.59.799} {\bibfield  {journal} {\bibinfo  {journal} {Phys. Rev. Lett.}\ }\textbf {\bibinfo {volume} {59}},\ \bibinfo {pages} {799} (\bibinfo {year} {1987})}\BibitemShut {NoStop}%
\bibitem [{\citenamefont {Pollmann}\ \emph {et~al.}(2010)\citenamefont {Pollmann}, \citenamefont {Turner}, \citenamefont {Berg},\ and\ \citenamefont {Oshikawa}}]{Pollmann_2010}%
  \BibitemOpen
  \bibfield  {author} {\bibinfo {author} {\bibfnamefont {F.}~\bibnamefont {Pollmann}}, \bibinfo {author} {\bibfnamefont {A.~M.}\ \bibnamefont {Turner}}, \bibinfo {author} {\bibfnamefont {E.}~\bibnamefont {Berg}},\ and\ \bibinfo {author} {\bibfnamefont {M.}~\bibnamefont {Oshikawa}},\ }\bibfield  {title} {\bibinfo {title} {Entanglement spectrum of a topological phase in one dimension},\ }\href {https://doi.org/10.1103/PhysRevB.81.064439} {\bibfield  {journal} {\bibinfo  {journal} {Phys. Rev. B}\ }\textbf {\bibinfo {volume} {81}},\ \bibinfo {pages} {064439} (\bibinfo {year} {2010})}\BibitemShut {NoStop}%
\bibitem [{\citenamefont {Stoudenmire}\ and\ \citenamefont {White}(2012)}]{2012_Stoudenmire}%
  \BibitemOpen
  \bibfield  {author} {\bibinfo {author} {\bibfnamefont {E.}~\bibnamefont {Stoudenmire}}\ and\ \bibinfo {author} {\bibfnamefont {S.~R.}\ \bibnamefont {White}},\ }\bibfield  {title} {\bibinfo {title} {Studying two-dimensional systems with the density matrix renormalization group},\ }\href {https://doi.org/10.1146/annurev-conmatphys-020911-125018} {\bibfield  {journal} {\bibinfo  {journal} {Annu. Rev. Cond. Matt. Phys.}\ }\textbf {\bibinfo {volume} {3}},\ \bibinfo {pages} {111} (\bibinfo {year} {2012})}\BibitemShut {NoStop}%
\bibitem [{\citenamefont {Ba{\~{n}}uls}\ \emph {et~al.}(2013)\citenamefont {Ba{\~{n}}uls}, \citenamefont {Cichy}, \citenamefont {Cirac},\ and\ \citenamefont {Jansen}}]{Banuls2013}%
  \BibitemOpen
  \bibfield  {author} {\bibinfo {author} {\bibfnamefont {M.~C.}\ \bibnamefont {Ba{\~{n}}uls}}, \bibinfo {author} {\bibfnamefont {K.}~\bibnamefont {Cichy}}, \bibinfo {author} {\bibfnamefont {J.~I.}\ \bibnamefont {Cirac}},\ and\ \bibinfo {author} {\bibfnamefont {K.}~\bibnamefont {Jansen}},\ }\bibfield  {title} {\bibinfo {title} {{The mass spectrum of the Schwinger model with matrix product states}},\ }\href {https://doi.org/10.1007/JHEP11(2013)158} {\bibfield  {journal} {\bibinfo  {journal} {J. High Energ. Phys.}\ }\textbf {\bibinfo {volume} {2013}}\bibinfo  {number} { (11)},\ \bibinfo {pages} {158}}\BibitemShut {NoStop}%
\bibitem [{\citenamefont {Haegeman}\ \emph {et~al.}(2012)\citenamefont {Haegeman}, \citenamefont {Pirvu}, \citenamefont {Weir}, \citenamefont {Cirac}, \citenamefont {Osborne}, \citenamefont {Verschelde},\ and\ \citenamefont {Verstraete}}]{2012_Haegeman}%
  \BibitemOpen
\bibfield  {number} {  }\bibfield  {author} {\bibinfo {author} {\bibfnamefont {J.}~\bibnamefont {Haegeman}}, \bibinfo {author} {\bibfnamefont {B.}~\bibnamefont {Pirvu}}, \bibinfo {author} {\bibfnamefont {D.~J.}\ \bibnamefont {Weir}}, \bibinfo {author} {\bibfnamefont {J.~I.}\ \bibnamefont {Cirac}}, \bibinfo {author} {\bibfnamefont {T.~J.}\ \bibnamefont {Osborne}}, \bibinfo {author} {\bibfnamefont {H.}~\bibnamefont {Verschelde}},\ and\ \bibinfo {author} {\bibfnamefont {F.}~\bibnamefont {Verstraete}},\ }\bibfield  {title} {\bibinfo {title} {Variational matrix product ansatz for dispersion relations},\ }\href {https://doi.org/10.1103/PhysRevB.85.100408} {\bibfield  {journal} {\bibinfo  {journal} {Phys. Rev. B}\ }\textbf {\bibinfo {volume} {85}},\ \bibinfo {pages} {100408} (\bibinfo {year} {2012})}\BibitemShut {NoStop}%
\bibitem [{\citenamefont {Haegeman}\ \emph {et~al.}(2013)\citenamefont {Haegeman}, \citenamefont {Osborne},\ and\ \citenamefont {Verstraete}}]{2013_Haegeman}%
  \BibitemOpen
  \bibfield  {author} {\bibinfo {author} {\bibfnamefont {J.}~\bibnamefont {Haegeman}}, \bibinfo {author} {\bibfnamefont {T.~J.}\ \bibnamefont {Osborne}},\ and\ \bibinfo {author} {\bibfnamefont {F.}~\bibnamefont {Verstraete}},\ }\bibfield  {title} {\bibinfo {title} {Post-matrix product state methods: To tangent space and beyond},\ }\href {https://doi.org/10.1103/PhysRevB.88.075133} {\bibfield  {journal} {\bibinfo  {journal} {Phys. Rev. B}\ }\textbf {\bibinfo {volume} {88}},\ \bibinfo {pages} {075133} (\bibinfo {year} {2013})}\BibitemShut {NoStop}%
\bibitem [{\citenamefont {Vanderstraeten}\ \emph {et~al.}(2018)\citenamefont {Vanderstraeten}, \citenamefont {Van~Damme}, \citenamefont {B\"uchler},\ and\ \citenamefont {Verstraete}}]{2018_Vanderstraeten}%
  \BibitemOpen
  \bibfield  {author} {\bibinfo {author} {\bibfnamefont {L.}~\bibnamefont {Vanderstraeten}}, \bibinfo {author} {\bibfnamefont {M.}~\bibnamefont {Van~Damme}}, \bibinfo {author} {\bibfnamefont {H.~P.}\ \bibnamefont {B\"uchler}},\ and\ \bibinfo {author} {\bibfnamefont {F.}~\bibnamefont {Verstraete}},\ }\bibfield  {title} {\bibinfo {title} {Quasiparticles in quantum spin chains with long-range interactions},\ }\href {https://doi.org/10.1103/PhysRevLett.121.090603} {\bibfield  {journal} {\bibinfo  {journal} {Phys. Rev. Lett.}\ }\textbf {\bibinfo {volume} {121}},\ \bibinfo {pages} {090603} (\bibinfo {year} {2018})}\BibitemShut {NoStop}%
\bibitem [{\citenamefont {Vanderstraeten}\ \emph {et~al.}(2019{\natexlab{b}})\citenamefont {Vanderstraeten}, \citenamefont {Haegeman},\ and\ \citenamefont {Verstraete}}]{2019_Vanderstraeten}%
  \BibitemOpen
  \bibfield  {author} {\bibinfo {author} {\bibfnamefont {L.}~\bibnamefont {Vanderstraeten}}, \bibinfo {author} {\bibfnamefont {J.}~\bibnamefont {Haegeman}},\ and\ \bibinfo {author} {\bibfnamefont {F.}~\bibnamefont {Verstraete}},\ }\bibfield  {title} {\bibinfo {title} {{Tangent-space methods for uniform matrix product states}},\ }\href {https://doi.org/10.21468/SciPostPhysLectNotes.7} {\bibfield  {journal} {\bibinfo  {journal} {SciPost Phys. Lect. Notes}\ }\textbf {\bibinfo {volume} {7}},\ \bibinfo {pages} {1} (\bibinfo {year} {2019}{\natexlab{b}})}\BibitemShut {NoStop}%
\bibitem [{\citenamefont {Draxler}\ \emph {et~al.}(2013)\citenamefont {Draxler}, \citenamefont {Haegeman}, \citenamefont {Osborne}, \citenamefont {Stojevic}, \citenamefont {Vanderstraeten},\ and\ \citenamefont {Verstraete}}]{2013_Draxler}%
  \BibitemOpen
  \bibfield  {author} {\bibinfo {author} {\bibfnamefont {D.}~\bibnamefont {Draxler}}, \bibinfo {author} {\bibfnamefont {J.}~\bibnamefont {Haegeman}}, \bibinfo {author} {\bibfnamefont {T.~J.}\ \bibnamefont {Osborne}}, \bibinfo {author} {\bibfnamefont {V.}~\bibnamefont {Stojevic}}, \bibinfo {author} {\bibfnamefont {L.}~\bibnamefont {Vanderstraeten}},\ and\ \bibinfo {author} {\bibfnamefont {F.}~\bibnamefont {Verstraete}},\ }\bibfield  {title} {\bibinfo {title} {Particles, holes, and solitons: A matrix product state approach},\ }\href {https://doi.org/10.1103/PhysRevLett.111.020402} {\bibfield  {journal} {\bibinfo  {journal} {Phys. Rev. Lett.}\ }\textbf {\bibinfo {volume} {111}},\ \bibinfo {pages} {020402} (\bibinfo {year} {2013})}\BibitemShut {NoStop}%
\bibitem [{\citenamefont {Vanderstraeten}\ \emph {et~al.}(2020)\citenamefont {Vanderstraeten}, \citenamefont {Wybo}, \citenamefont {Chepiga}, \citenamefont {Verstraete},\ and\ \citenamefont {Mila}}]{2020_Vanderstraeten_spinons}%
  \BibitemOpen
  \bibfield  {author} {\bibinfo {author} {\bibfnamefont {L.}~\bibnamefont {Vanderstraeten}}, \bibinfo {author} {\bibfnamefont {E.}~\bibnamefont {Wybo}}, \bibinfo {author} {\bibfnamefont {N.}~\bibnamefont {Chepiga}}, \bibinfo {author} {\bibfnamefont {F.}~\bibnamefont {Verstraete}},\ and\ \bibinfo {author} {\bibfnamefont {F.}~\bibnamefont {Mila}},\ }\bibfield  {title} {\bibinfo {title} {Spinon confinement and deconfinement in spin-1 chains},\ }\href {https://doi.org/10.1103/PhysRevB.101.115138} {\bibfield  {journal} {\bibinfo  {journal} {Phys. Rev. B}\ }\textbf {\bibinfo {volume} {101}},\ \bibinfo {pages} {115138} (\bibinfo {year} {2020})}\BibitemShut {NoStop}%
\bibitem [{\citenamefont {Vanderstraeten}\ \emph {et~al.}(2015{\natexlab{a}})\citenamefont {Vanderstraeten}, \citenamefont {Verstraete},\ and\ \citenamefont {Haegeman}}]{2015_Vanderstraeten_bound}%
  \BibitemOpen
  \bibfield  {author} {\bibinfo {author} {\bibfnamefont {L.}~\bibnamefont {Vanderstraeten}}, \bibinfo {author} {\bibfnamefont {F.}~\bibnamefont {Verstraete}},\ and\ \bibinfo {author} {\bibfnamefont {J.}~\bibnamefont {Haegeman}},\ }\bibfield  {title} {\bibinfo {title} {Scattering particles in quantum spin chains},\ }\href {https://doi.org/10.1103/PhysRevB.92.125136} {\bibfield  {journal} {\bibinfo  {journal} {Phys. Rev. B}\ }\textbf {\bibinfo {volume} {92}},\ \bibinfo {pages} {125136} (\bibinfo {year} {2015}{\natexlab{a}})}\BibitemShut {NoStop}%
\bibitem [{\citenamefont {Vanderstraeten}\ \emph {et~al.}(2014)\citenamefont {Vanderstraeten}, \citenamefont {Haegeman}, \citenamefont {Osborne},\ and\ \citenamefont {Verstraete}}]{2014_Vanderstraeten}%
  \BibitemOpen
  \bibfield  {author} {\bibinfo {author} {\bibfnamefont {L.}~\bibnamefont {Vanderstraeten}}, \bibinfo {author} {\bibfnamefont {J.}~\bibnamefont {Haegeman}}, \bibinfo {author} {\bibfnamefont {T.~J.}\ \bibnamefont {Osborne}},\ and\ \bibinfo {author} {\bibfnamefont {F.}~\bibnamefont {Verstraete}},\ }\bibfield  {title} {\bibinfo {title} {$s$ matrix from matrix product states},\ }\href {https://doi.org/10.1103/PhysRevLett.112.257202} {\bibfield  {journal} {\bibinfo  {journal} {Phys. Rev. Lett.}\ }\textbf {\bibinfo {volume} {112}},\ \bibinfo {pages} {257202} (\bibinfo {year} {2014})}\BibitemShut {NoStop}%
\bibitem [{\citenamefont {Zauner-Stauber}\ \emph {et~al.}(2018)\citenamefont {Zauner-Stauber}, \citenamefont {Vanderstraeten}, \citenamefont {Haegeman}, \citenamefont {McCulloch},\ and\ \citenamefont {Verstraete}}]{2018_Zauner}%
  \BibitemOpen
  \bibfield  {author} {\bibinfo {author} {\bibfnamefont {V.}~\bibnamefont {Zauner-Stauber}}, \bibinfo {author} {\bibfnamefont {L.}~\bibnamefont {Vanderstraeten}}, \bibinfo {author} {\bibfnamefont {J.}~\bibnamefont {Haegeman}}, \bibinfo {author} {\bibfnamefont {I.~P.}\ \bibnamefont {McCulloch}},\ and\ \bibinfo {author} {\bibfnamefont {F.}~\bibnamefont {Verstraete}},\ }\bibfield  {title} {\bibinfo {title} {Topological nature of spinons and holons: Elementary excitations from matrix product states with conserved symmetries},\ }\href {https://doi.org/10.1103/PhysRevB.97.235155} {\bibfield  {journal} {\bibinfo  {journal} {Phys. Rev. B}\ }\textbf {\bibinfo {volume} {97}},\ \bibinfo {pages} {235155} (\bibinfo {year} {2018})}\BibitemShut {NoStop}%
\bibitem [{\citenamefont {Van~Damme}\ \emph {et~al.}(2021)\citenamefont {Van~Damme}, \citenamefont {Vanhove}, \citenamefont {Haegeman}, \citenamefont {Verstraete},\ and\ \citenamefont {Vanderstraeten}}]{2021_Damme}%
  \BibitemOpen
  \bibfield  {author} {\bibinfo {author} {\bibfnamefont {M.}~\bibnamefont {Van~Damme}}, \bibinfo {author} {\bibfnamefont {R.}~\bibnamefont {Vanhove}}, \bibinfo {author} {\bibfnamefont {J.}~\bibnamefont {Haegeman}}, \bibinfo {author} {\bibfnamefont {F.}~\bibnamefont {Verstraete}},\ and\ \bibinfo {author} {\bibfnamefont {L.}~\bibnamefont {Vanderstraeten}},\ }\bibfield  {title} {\bibinfo {title} {Efficient matrix product state methods for extracting spectral information on rings and cylinders},\ }\href {https://doi.org/10.1103/PhysRevB.104.115142} {\bibfield  {journal} {\bibinfo  {journal} {Phys. Rev. B}\ }\textbf {\bibinfo {volume} {104}},\ \bibinfo {pages} {115142} (\bibinfo {year} {2021})}\BibitemShut {NoStop}%
\bibitem [{\citenamefont {Zhang}\ \emph {et~al.}(2024)\citenamefont {Zhang}, \citenamefont {Chi}, \citenamefont {Liu},\ and\ \citenamefont {Wang}}]{zhang2023}%
  \BibitemOpen
  \bibfield  {author} {\bibinfo {author} {\bibfnamefont {X.-Y.}\ \bibnamefont {Zhang}}, \bibinfo {author} {\bibfnamefont {R.}~\bibnamefont {Chi}}, \bibinfo {author} {\bibfnamefont {Y.}~\bibnamefont {Liu}},\ and\ \bibinfo {author} {\bibfnamefont {L.}~\bibnamefont {Wang}},\ }\bibfield  {title} {\bibinfo {title} {Two-dimensional excitation information by a matrix product state method on infinite helices},\ }\href {https://doi.org/10.1103/PhysRevB.109.075129} {\bibfield  {journal} {\bibinfo  {journal} {Phys. Rev. B}\ }\textbf {\bibinfo {volume} {109}},\ \bibinfo {pages} {075129} (\bibinfo {year} {2024})}\BibitemShut {NoStop}%
\bibitem [{\citenamefont {Vanderstraeten}\ \emph {et~al.}(2015{\natexlab{b}})\citenamefont {Vanderstraeten}, \citenamefont {Mari\"en}, \citenamefont {Verstraete},\ and\ \citenamefont {Haegeman}}]{2015_Vanderstraeten}%
  \BibitemOpen
  \bibfield  {author} {\bibinfo {author} {\bibfnamefont {L.}~\bibnamefont {Vanderstraeten}}, \bibinfo {author} {\bibfnamefont {M.}~\bibnamefont {Mari\"en}}, \bibinfo {author} {\bibfnamefont {F.}~\bibnamefont {Verstraete}},\ and\ \bibinfo {author} {\bibfnamefont {J.}~\bibnamefont {Haegeman}},\ }\bibfield  {title} {\bibinfo {title} {Excitations and the tangent space of projected entangled-pair states},\ }\href {https://doi.org/10.1103/PhysRevB.92.201111} {\bibfield  {journal} {\bibinfo  {journal} {Phys. Rev. B}\ }\textbf {\bibinfo {volume} {92}},\ \bibinfo {pages} {201111} (\bibinfo {year} {2015}{\natexlab{b}})}\BibitemShut {NoStop}%
\bibitem [{\citenamefont {Ponsioen}\ and\ \citenamefont {Corboz}(2020)}]{2020_Ponsioen}%
  \BibitemOpen
  \bibfield  {author} {\bibinfo {author} {\bibfnamefont {B.}~\bibnamefont {Ponsioen}}\ and\ \bibinfo {author} {\bibfnamefont {P.}~\bibnamefont {Corboz}},\ }\bibfield  {title} {\bibinfo {title} {Excitations with projected entangled pair states using the corner transfer matrix method},\ }\href {https://doi.org/10.1103/PhysRevB.101.195109} {\bibfield  {journal} {\bibinfo  {journal} {Phys. Rev. B}\ }\textbf {\bibinfo {volume} {101}},\ \bibinfo {pages} {195109} (\bibinfo {year} {2020})}\BibitemShut {NoStop}%
\bibitem [{\citenamefont {Ponsioen}\ \emph {et~al.}(2022)\citenamefont {Ponsioen}, \citenamefont {Assaad},\ and\ \citenamefont {Corboz}}]{2022_Ponsioen}%
  \BibitemOpen
  \bibfield  {author} {\bibinfo {author} {\bibfnamefont {B.}~\bibnamefont {Ponsioen}}, \bibinfo {author} {\bibfnamefont {F.~F.}\ \bibnamefont {Assaad}},\ and\ \bibinfo {author} {\bibfnamefont {P.}~\bibnamefont {Corboz}},\ }\bibfield  {title} {\bibinfo {title} {{Automatic differentiation applied to excitations with projected entangled pair states}},\ }\href {https://doi.org/10.21468/SciPostPhys.12.1.006} {\bibfield  {journal} {\bibinfo  {journal} {SciPost Phys.}\ }\textbf {\bibinfo {volume} {12}},\ \bibinfo {pages} {006} (\bibinfo {year} {2022})}\BibitemShut {NoStop}%
\bibitem [{\citenamefont {Baker}\ \emph {et~al.}(2023)\citenamefont {Baker}, \citenamefont {Foley},\ and\ \citenamefont {Sénéchal}}]{baker2023direct}%
  \BibitemOpen
  \bibfield  {author} {\bibinfo {author} {\bibfnamefont {T.~E.}\ \bibnamefont {Baker}}, \bibinfo {author} {\bibfnamefont {A.}~\bibnamefont {Foley}},\ and\ \bibinfo {author} {\bibfnamefont {D.}~\bibnamefont {Sénéchal}},\ }\href@noop {} {\bibinfo {title} {{Direct solution of multiple excitations in a matrix product state with block Lanczos}}} (\bibinfo {year} {2023}),\ \Eprint {https://arxiv.org/abs/2109.08181} {arXiv:2109.08181 [cond-mat.str-el]} \BibitemShut {NoStop}%
\bibitem [{\citenamefont {Dolgov}\ \emph {et~al.}(2014)\citenamefont {Dolgov}, \citenamefont {Khoromskij}, \citenamefont {Oseledets},\ and\ \citenamefont {Savostyanov}}]{2014_Dolgov}%
  \BibitemOpen
  \bibfield  {author} {\bibinfo {author} {\bibfnamefont {S.}~\bibnamefont {Dolgov}}, \bibinfo {author} {\bibfnamefont {B.}~\bibnamefont {Khoromskij}}, \bibinfo {author} {\bibfnamefont {I.}~\bibnamefont {Oseledets}},\ and\ \bibinfo {author} {\bibfnamefont {D.}~\bibnamefont {Savostyanov}},\ }\bibfield  {title} {\bibinfo {title} {Computation of extreme eigenvalues in higher dimensions using block tensor train format},\ }\href {https://doi.org/https://doi.org/10.1016/j.cpc.2013.12.017} {\bibfield  {journal} {\bibinfo  {journal} {Comput. Phys. Commun.}\ }\textbf {\bibinfo {volume} {185}},\ \bibinfo {pages} {1207} (\bibinfo {year} {2014})}\BibitemShut {NoStop}%
\bibitem [{\citenamefont {Chepiga}\ and\ \citenamefont {Mila}(2017)}]{2017_Chepiga}%
  \BibitemOpen
  \bibfield  {author} {\bibinfo {author} {\bibfnamefont {N.}~\bibnamefont {Chepiga}}\ and\ \bibinfo {author} {\bibfnamefont {F.}~\bibnamefont {Mila}},\ }\bibfield  {title} {\bibinfo {title} {Excitation spectrum and density matrix renormalization group iterations},\ }\href {https://doi.org/10.1103/PhysRevB.96.054425} {\bibfield  {journal} {\bibinfo  {journal} {Phys. Rev. B}\ }\textbf {\bibinfo {volume} {96}},\ \bibinfo {pages} {054425} (\bibinfo {year} {2017})}\BibitemShut {NoStop}%
\bibitem [{\citenamefont {Zauner}\ \emph {et~al.}(2015)\citenamefont {Zauner}, \citenamefont {Draxler}, \citenamefont {Vanderstraeten}, \citenamefont {Degroote}, \citenamefont {Haegeman}, \citenamefont {Rams}, \citenamefont {Stojevic}, \citenamefont {Schuch},\ and\ \citenamefont {Verstraete}}]{2014_Zauner}%
  \BibitemOpen
  \bibfield  {author} {\bibinfo {author} {\bibfnamefont {V.}~\bibnamefont {Zauner}}, \bibinfo {author} {\bibfnamefont {D.}~\bibnamefont {Draxler}}, \bibinfo {author} {\bibfnamefont {L.}~\bibnamefont {Vanderstraeten}}, \bibinfo {author} {\bibfnamefont {M.}~\bibnamefont {Degroote}}, \bibinfo {author} {\bibfnamefont {J.}~\bibnamefont {Haegeman}}, \bibinfo {author} {\bibfnamefont {M.~M.}\ \bibnamefont {Rams}}, \bibinfo {author} {\bibfnamefont {V.}~\bibnamefont {Stojevic}}, \bibinfo {author} {\bibfnamefont {N.}~\bibnamefont {Schuch}},\ and\ \bibinfo {author} {\bibfnamefont {F.}~\bibnamefont {Verstraete}},\ }\bibfield  {title} {\bibinfo {title} {Transfer matrices and excitations with matrix product states},\ }\href {https://doi.org/10.1088/1367-2630/17/5/053002} {\bibfield  {journal} {\bibinfo  {journal} {New J. Phys.}\ }\textbf {\bibinfo {volume} {17}},\ \bibinfo {pages} {053002} (\bibinfo {year} {2015})}\BibitemShut {NoStop}%
\bibitem [{\citenamefont {Eberharter}\ \emph {et~al.}(2023)\citenamefont {Eberharter}, \citenamefont {Vanderstraeten}, \citenamefont {Verstraete},\ and\ \citenamefont {L\"auchli}}]{2023_Eberharter}%
  \BibitemOpen
  \bibfield  {author} {\bibinfo {author} {\bibfnamefont {A.~A.}\ \bibnamefont {Eberharter}}, \bibinfo {author} {\bibfnamefont {L.}~\bibnamefont {Vanderstraeten}}, \bibinfo {author} {\bibfnamefont {F.}~\bibnamefont {Verstraete}},\ and\ \bibinfo {author} {\bibfnamefont {A.~M.}\ \bibnamefont {L\"auchli}},\ }\bibfield  {title} {\bibinfo {title} {Extracting the speed of light from matrix product states},\ }\href {https://doi.org/10.1103/PhysRevLett.131.226502} {\bibfield  {journal} {\bibinfo  {journal} {Phys. Rev. Lett.}\ }\textbf {\bibinfo {volume} {131}},\ \bibinfo {pages} {226502} (\bibinfo {year} {2023})}\BibitemShut {NoStop}%
\bibitem [{\citenamefont {Garcia-Saez}\ \emph {et~al.}(2013)\citenamefont {Garcia-Saez}, \citenamefont {Murg},\ and\ \citenamefont {Wei}}]{2013_Garcia}%
  \BibitemOpen
  \bibfield  {author} {\bibinfo {author} {\bibfnamefont {A.}~\bibnamefont {Garcia-Saez}}, \bibinfo {author} {\bibfnamefont {V.}~\bibnamefont {Murg}},\ and\ \bibinfo {author} {\bibfnamefont {T.-C.}\ \bibnamefont {Wei}},\ }\bibfield  {title} {\bibinfo {title} {{Spectral gaps of Affleck-Kennedy-Lieb-Tasaki Hamiltonians using tensor network methods}},\ }\href {https://doi.org/10.1103/PhysRevB.88.245118} {\bibfield  {journal} {\bibinfo  {journal} {Phys. Rev. B}\ }\textbf {\bibinfo {volume} {88}},\ \bibinfo {pages} {245118} (\bibinfo {year} {2013})}\BibitemShut {NoStop}%
\bibitem [{\citenamefont {Khemani}\ \emph {et~al.}(2016)\citenamefont {Khemani}, \citenamefont {Pollmann},\ and\ \citenamefont {Sondhi}}]{2016_Khemani}%
  \BibitemOpen
  \bibfield  {author} {\bibinfo {author} {\bibfnamefont {V.}~\bibnamefont {Khemani}}, \bibinfo {author} {\bibfnamefont {F.}~\bibnamefont {Pollmann}},\ and\ \bibinfo {author} {\bibfnamefont {S.~L.}\ \bibnamefont {Sondhi}},\ }\bibfield  {title} {\bibinfo {title} {Obtaining highly excited eigenstates of many-body localized {H}amiltonians by the density matrix renormalization group approach},\ }\href {https://doi.org/10.1103/PhysRevLett.116.247204} {\bibfield  {journal} {\bibinfo  {journal} {Phys. Rev. Lett.}\ }\textbf {\bibinfo {volume} {116}},\ \bibinfo {pages} {247204} (\bibinfo {year} {2016})}\BibitemShut {NoStop}%
\bibitem [{\citenamefont {Pollmann}\ \emph {et~al.}(2016)\citenamefont {Pollmann}, \citenamefont {Khemani}, \citenamefont {Cirac},\ and\ \citenamefont {Sondhi}}]{2016_Pollmann}%
  \BibitemOpen
  \bibfield  {author} {\bibinfo {author} {\bibfnamefont {F.}~\bibnamefont {Pollmann}}, \bibinfo {author} {\bibfnamefont {V.}~\bibnamefont {Khemani}}, \bibinfo {author} {\bibfnamefont {J.~I.}\ \bibnamefont {Cirac}},\ and\ \bibinfo {author} {\bibfnamefont {S.~L.}\ \bibnamefont {Sondhi}},\ }\bibfield  {title} {\bibinfo {title} {{Efficient variational diagonalization of fully many-body localized Hamiltonians}},\ }\href {https://doi.org/10.1103/PhysRevB.94.041116} {\bibfield  {journal} {\bibinfo  {journal} {Phys. Rev. B}\ }\textbf {\bibinfo {volume} {94}},\ \bibinfo {pages} {041116} (\bibinfo {year} {2016})}\BibitemShut {NoStop}%
\bibitem [{\citenamefont {Phien}\ \emph {et~al.}(2015)\citenamefont {Phien}, \citenamefont {McCulloch},\ and\ \citenamefont {Vidal}}]{2015_Phien}%
  \BibitemOpen
  \bibfield  {author} {\bibinfo {author} {\bibfnamefont {H.~N.}\ \bibnamefont {Phien}}, \bibinfo {author} {\bibfnamefont {I.~P.}\ \bibnamefont {McCulloch}},\ and\ \bibinfo {author} {\bibfnamefont {G.}~\bibnamefont {Vidal}},\ }\bibfield  {title} {\bibinfo {title} {Fast convergence of imaginary time evolution tensor network algorithms by recycling the environment},\ }\href {https://doi.org/10.1103/PhysRevB.91.115137} {\bibfield  {journal} {\bibinfo  {journal} {Phys. Rev. B}\ }\textbf {\bibinfo {volume} {91}},\ \bibinfo {pages} {115137} (\bibinfo {year} {2015})}\BibitemShut {NoStop}%
\bibitem [{\citenamefont {Ran}\ \emph {et~al.}(2012)\citenamefont {Ran}, \citenamefont {Li}, \citenamefont {Xi}, \citenamefont {Zhang},\ and\ \citenamefont {Su}}]{2012_Ran}%
  \BibitemOpen
  \bibfield  {author} {\bibinfo {author} {\bibfnamefont {S.-J.}\ \bibnamefont {Ran}}, \bibinfo {author} {\bibfnamefont {W.}~\bibnamefont {Li}}, \bibinfo {author} {\bibfnamefont {B.}~\bibnamefont {Xi}}, \bibinfo {author} {\bibfnamefont {Z.}~\bibnamefont {Zhang}},\ and\ \bibinfo {author} {\bibfnamefont {G.}~\bibnamefont {Su}},\ }\bibfield  {title} {\bibinfo {title} {Optimized decimation of tensor networks with super-orthogonalization for two-dimensional quantum lattice models},\ }\href {https://doi.org/10.1103/PhysRevB.86.134429} {\bibfield  {journal} {\bibinfo  {journal} {Phys. Rev. B}\ }\textbf {\bibinfo {volume} {86}},\ \bibinfo {pages} {134429} (\bibinfo {year} {2012})}\BibitemShut {NoStop}%
\bibitem [{\citenamefont {Zaletel}\ \emph {et~al.}(2015)\citenamefont {Zaletel}, \citenamefont {Mong}, \citenamefont {Karrasch}, \citenamefont {Moore},\ and\ \citenamefont {Pollmann}}]{2015_Zaletel}%
  \BibitemOpen
  \bibfield  {author} {\bibinfo {author} {\bibfnamefont {M.~P.}\ \bibnamefont {Zaletel}}, \bibinfo {author} {\bibfnamefont {R.~S.~K.}\ \bibnamefont {Mong}}, \bibinfo {author} {\bibfnamefont {C.}~\bibnamefont {Karrasch}}, \bibinfo {author} {\bibfnamefont {J.~E.}\ \bibnamefont {Moore}},\ and\ \bibinfo {author} {\bibfnamefont {F.}~\bibnamefont {Pollmann}},\ }\bibfield  {title} {\bibinfo {title} {Time-evolving a matrix product state with long-ranged interactions},\ }\href {https://doi.org/10.1103/PhysRevB.91.165112} {\bibfield  {journal} {\bibinfo  {journal} {Phys. Rev. B}\ }\textbf {\bibinfo {volume} {91}},\ \bibinfo {pages} {165112} (\bibinfo {year} {2015})}\BibitemShut {NoStop}%
\bibitem [{\citenamefont {Baxter}(1968)}]{baxter1968dimers}%
  \BibitemOpen
  \bibfield  {author} {\bibinfo {author} {\bibfnamefont {R.~J.}\ \bibnamefont {Baxter}},\ }\bibfield  {title} {\bibinfo {title} {Dimers on a rectangular lattice},\ }\href@noop {} {\bibfield  {journal} {\bibinfo  {journal} {Journal of Mathematical Physics}\ }\textbf {\bibinfo {volume} {9}},\ \bibinfo {pages} {650} (\bibinfo {year} {1968})}\BibitemShut {NoStop}%
\bibitem [{\citenamefont {Nishino}\ and\ \citenamefont {Okunishi}(1996)}]{1996_Nishino}%
  \BibitemOpen
  \bibfield  {author} {\bibinfo {author} {\bibfnamefont {T.}~\bibnamefont {Nishino}}\ and\ \bibinfo {author} {\bibfnamefont {K.}~\bibnamefont {Okunishi}},\ }\bibfield  {title} {\bibinfo {title} {Corner transfer matrix renormalization group method},\ }\href {https://doi.org/10.1143/JPSJ.65.891} {\bibfield  {journal} {\bibinfo  {journal} {J. Phys. Soc. Jpn.}\ }\textbf {\bibinfo {volume} {65}},\ \bibinfo {pages} {891} (\bibinfo {year} {1996})}\BibitemShut {NoStop}%
\bibitem [{\citenamefont {Or\'us}\ and\ \citenamefont {Vidal}(2009)}]{2009_Orus}%
  \BibitemOpen
  \bibfield  {author} {\bibinfo {author} {\bibfnamefont {R.}~\bibnamefont {Or\'us}}\ and\ \bibinfo {author} {\bibfnamefont {G.}~\bibnamefont {Vidal}},\ }\bibfield  {title} {\bibinfo {title} {Simulation of two-dimensional quantum systems on an infinite lattice revisited: Corner transfer matrix for tensor contraction},\ }\href {https://doi.org/10.1103/PhysRevB.80.094403} {\bibfield  {journal} {\bibinfo  {journal} {Phys. Rev. B}\ }\textbf {\bibinfo {volume} {80}},\ \bibinfo {pages} {094403} (\bibinfo {year} {2009})}\BibitemShut {NoStop}%
\bibitem [{\citenamefont {Nishino}\ and\ \citenamefont {Okunishi}(1997)}]{Nishino_1997}%
  \BibitemOpen
  \bibfield  {author} {\bibinfo {author} {\bibfnamefont {T.}~\bibnamefont {Nishino}}\ and\ \bibinfo {author} {\bibfnamefont {K.}~\bibnamefont {Okunishi}},\ }\bibfield  {title} {\bibinfo {title} {Corner transfer matrix algorithm for classical renormalization group},\ }\href {https://doi.org/10.1143/JPSJ.66.3040} {\bibfield  {journal} {\bibinfo  {journal} {J. Phys. Soc. Jpn.}\ }\textbf {\bibinfo {volume} {66}},\ \bibinfo {pages} {3040} (\bibinfo {year} {1997})}\BibitemShut {NoStop}%
\bibitem [{\citenamefont {Patra}\ \emph {et~al.}(2024)\citenamefont {Patra}, \citenamefont {Jahromi}, \citenamefont {Singh},\ and\ \citenamefont {Or\'us}}]{2023_Patra_IBM}%
  \BibitemOpen
  \bibfield  {author} {\bibinfo {author} {\bibfnamefont {S.}~\bibnamefont {Patra}}, \bibinfo {author} {\bibfnamefont {S.~S.}\ \bibnamefont {Jahromi}}, \bibinfo {author} {\bibfnamefont {S.}~\bibnamefont {Singh}},\ and\ \bibinfo {author} {\bibfnamefont {R.}~\bibnamefont {Or\'us}},\ }\bibfield  {title} {\bibinfo {title} {Efficient tensor network simulation of {IBM}'s largest quantum processors},\ }\href {https://doi.org/10.1103/PhysRevResearch.6.013326} {\bibfield  {journal} {\bibinfo  {journal} {Phys. Rev. Res.}\ }\textbf {\bibinfo {volume} {6}},\ \bibinfo {pages} {013326} (\bibinfo {year} {2024})}\BibitemShut {NoStop}%
\bibitem [{\citenamefont {Tindall}\ \emph {et~al.}(2024)\citenamefont {Tindall}, \citenamefont {Fishman}, \citenamefont {Stoudenmire},\ and\ \citenamefont {Sels}}]{2023_tindall_IBM}%
  \BibitemOpen
  \bibfield  {author} {\bibinfo {author} {\bibfnamefont {J.}~\bibnamefont {Tindall}}, \bibinfo {author} {\bibfnamefont {M.}~\bibnamefont {Fishman}}, \bibinfo {author} {\bibfnamefont {E.~M.}\ \bibnamefont {Stoudenmire}},\ and\ \bibinfo {author} {\bibfnamefont {D.}~\bibnamefont {Sels}},\ }\bibfield  {title} {\bibinfo {title} {Efficient tensor network simulation of {IBM}'s eagle kicked {I}sing experiment},\ }\href {https://doi.org/10.1103/PRXQuantum.5.010308} {\bibfield  {journal} {\bibinfo  {journal} {PRX Quant.}\ }\textbf {\bibinfo {volume} {5}},\ \bibinfo {pages} {010308} (\bibinfo {year} {2024})}\BibitemShut {NoStop}%
\bibitem [{\citenamefont {Jordan}\ \emph {et~al.}(2008)\citenamefont {Jordan}, \citenamefont {Or\'us}, \citenamefont {Vidal}, \citenamefont {Verstraete},\ and\ \citenamefont {Cirac}}]{2008_Jordan}%
  \BibitemOpen
  \bibfield  {author} {\bibinfo {author} {\bibfnamefont {J.}~\bibnamefont {Jordan}}, \bibinfo {author} {\bibfnamefont {R.}~\bibnamefont {Or\'us}}, \bibinfo {author} {\bibfnamefont {G.}~\bibnamefont {Vidal}}, \bibinfo {author} {\bibfnamefont {F.}~\bibnamefont {Verstraete}},\ and\ \bibinfo {author} {\bibfnamefont {J.~I.}\ \bibnamefont {Cirac}},\ }\bibfield  {title} {\bibinfo {title} {Classical simulation of infinite-size quantum lattice systems in two spatial dimensions},\ }\href {https://doi.org/10.1103/PhysRevLett.101.250602} {\bibfield  {journal} {\bibinfo  {journal} {Phys. Rev. Lett.}\ }\textbf {\bibinfo {volume} {101}},\ \bibinfo {pages} {250602} (\bibinfo {year} {2008})}\BibitemShut {NoStop}%
\bibitem [{\citenamefont {Dziarmaga}(2022)}]{2022_Dzarmaga}%
  \BibitemOpen
  \bibfield  {author} {\bibinfo {author} {\bibfnamefont {J.}~\bibnamefont {Dziarmaga}},\ }\bibfield  {title} {\bibinfo {title} {Time evolution of an infinite projected entangled pair state: A gradient tensor update in the tangent space},\ }\href {https://doi.org/10.1103/PhysRevB.106.014304} {\bibfield  {journal} {\bibinfo  {journal} {Phys. Rev. B}\ }\textbf {\bibinfo {volume} {106}},\ \bibinfo {pages} {014304} (\bibinfo {year} {2022})}\BibitemShut {NoStop}%
\bibitem [{\citenamefont {Braiorr-Orrs}\ \emph {et~al.}(2016)\citenamefont {Braiorr-Orrs}, \citenamefont {Weyrauch},\ and\ \citenamefont {Rakov}}]{braiorr2016phase}%
  \BibitemOpen
  \bibfield  {author} {\bibinfo {author} {\bibfnamefont {B.}~\bibnamefont {Braiorr-Orrs}}, \bibinfo {author} {\bibfnamefont {M.}~\bibnamefont {Weyrauch}},\ and\ \bibinfo {author} {\bibfnamefont {M.~V.}\ \bibnamefont {Rakov}},\ }\bibfield  {title} {\bibinfo {title} {Phase diagrams of one-, two-, and three-dimensional quantum spin systems derived from entanglement properties},\ }\href {https://doi.org/10.26421/QIC16.9-10-9} {\bibfield  {journal} {\bibinfo  {journal} {Quant. Inf. Comput.}\ }\textbf {\bibinfo {volume} {16}},\ \bibinfo {pages} {885–899} (\bibinfo {year} {2016})}\BibitemShut {NoStop}%
\bibitem [{\citenamefont {Jahromi}\ \emph {et~al.}(2021)\citenamefont {Jahromi}, \citenamefont {Yarloo},\ and\ \citenamefont {Or\'us}}]{2021_Jahromi}%
  \BibitemOpen
  \bibfield  {author} {\bibinfo {author} {\bibfnamefont {S.~S.}\ \bibnamefont {Jahromi}}, \bibinfo {author} {\bibfnamefont {H.}~\bibnamefont {Yarloo}},\ and\ \bibinfo {author} {\bibfnamefont {R.}~\bibnamefont {Or\'us}},\ }\bibfield  {title} {\bibinfo {title} {{Thermodynamics of three-dimensional Kitaev quantum spin liquids via tensor networks}},\ }\href {https://doi.org/10.1103/PhysRevResearch.3.033205} {\bibfield  {journal} {\bibinfo  {journal} {Phys. Rev. Res.}\ }\textbf {\bibinfo {volume} {3}},\ \bibinfo {pages} {033205} (\bibinfo {year} {2021})}\BibitemShut {NoStop}%
\bibitem [{\citenamefont {Vlaar}\ and\ \citenamefont {Corboz}(2021)}]{2021_Vlaar}%
  \BibitemOpen
  \bibfield  {author} {\bibinfo {author} {\bibfnamefont {P.~C.~G.}\ \bibnamefont {Vlaar}}\ and\ \bibinfo {author} {\bibfnamefont {P.}~\bibnamefont {Corboz}},\ }\bibfield  {title} {\bibinfo {title} {Simulation of three-dimensional quantum systems with projected entangled-pair states},\ }\href {https://doi.org/10.1103/PhysRevB.103.205137} {\bibfield  {journal} {\bibinfo  {journal} {Phys. Rev. B}\ }\textbf {\bibinfo {volume} {103}},\ \bibinfo {pages} {205137} (\bibinfo {year} {2021})}\BibitemShut {NoStop}%
\bibitem [{\citenamefont {Klus}\ \emph {et~al.}(2018)\citenamefont {Klus}, \citenamefont {Gel{\ss}}, \citenamefont {Peitz},\ and\ \citenamefont {Sch{\"u}tte}}]{klus2018tensor}%
  \BibitemOpen
  \bibfield  {author} {\bibinfo {author} {\bibfnamefont {S.}~\bibnamefont {Klus}}, \bibinfo {author} {\bibfnamefont {P.}~\bibnamefont {Gel{\ss}}}, \bibinfo {author} {\bibfnamefont {S.}~\bibnamefont {Peitz}},\ and\ \bibinfo {author} {\bibfnamefont {C.}~\bibnamefont {Sch{\"u}tte}},\ }\bibfield  {title} {\bibinfo {title} {Tensor-based dynamic mode decomposition},\ }\href {https://doi.org/10.1088/1361-6544/aabc8f} {\bibfield  {journal} {\bibinfo  {journal} {Nonlinearity}\ }\textbf {\bibinfo {volume} {31}},\ \bibinfo {pages} {3359} (\bibinfo {year} {2018})}\BibitemShut {NoStop}%
\bibitem [{\citenamefont {Gel{\ss}}\ and\ \citenamefont {Sch{\"u}tte}(2018)}]{gelss2018tensor}%
  \BibitemOpen
  \bibfield  {author} {\bibinfo {author} {\bibfnamefont {P.}~\bibnamefont {Gel{\ss}}}\ and\ \bibinfo {author} {\bibfnamefont {C.}~\bibnamefont {Sch{\"u}tte}},\ }\href@noop {} {\bibinfo {title} {Tensor-generated fractals: Using tensor decompositions for creating self-similar patterns}} (\bibinfo {year} {2018}),\ \Eprint {https://arxiv.org/abs/1812.00814} {arXiv:1812.00814 [math.GM]} \BibitemShut {NoStop}%
\bibitem [{\citenamefont {Klus}\ and\ \citenamefont {Gel{\ss}}(2019)}]{klus2019tensor}%
  \BibitemOpen
  \bibfield  {author} {\bibinfo {author} {\bibfnamefont {S.}~\bibnamefont {Klus}}\ and\ \bibinfo {author} {\bibfnamefont {P.}~\bibnamefont {Gel{\ss}}},\ }\bibfield  {title} {\bibinfo {title} {Tensor-based algorithms for image classification},\ }\href {https://doi.org/doi.org/10.3390/a12110240} {\bibfield  {journal} {\bibinfo  {journal} {Algorithms}\ }\textbf {\bibinfo {volume} {12}},\ \bibinfo {pages} {240} (\bibinfo {year} {2019})}\BibitemShut {NoStop}%
\bibitem [{\citenamefont {L{\"u}cke}\ and\ \citenamefont {N{\"u}ske}(2022)}]{lucke2022tgedmd}%
  \BibitemOpen
  \bibfield  {author} {\bibinfo {author} {\bibfnamefont {M.}~\bibnamefont {L{\"u}cke}}\ and\ \bibinfo {author} {\bibfnamefont {F.}~\bibnamefont {N{\"u}ske}},\ }\bibfield  {title} {\bibinfo {title} {{tgEDMD: Approximation of the Kolmogorov operator in tensor train format}},\ }\href {https://doi.org/10.1007/s00332-022-09801-0} {\bibfield  {journal} {\bibinfo  {journal} {J. Nonlin. Sci.}\ }\textbf {\bibinfo {volume} {32}},\ \bibinfo {pages} {44} (\bibinfo {year} {2022})}\BibitemShut {NoStop}%
\bibitem [{Note1()}]{Note1}%
  \BibitemOpen
  \bibinfo {note} {\protect \href {https://www.energy.gov/downloads/doe-public-access-plan}{https://www.energy.gov/downloads/doe-public-access-plan}}\BibitemShut {NoStop}%
\end{thebibliography}%

\end{document}